\documentclass[11pt]{article}

\usepackage{amsmath,amsfonts,amssymb, mathrsfs}

\usepackage[paper=letterpaper,margin=1in]{geometry}
\usepackage[T1]{fontenc}
\usepackage{microtype}

\usepackage{amsthm,etoolbox,changepage}

\usepackage{hyperref}
\usepackage{xcolor}
\definecolor{darkred}{rgb}{0.5, 0, 0}
\definecolor{darkgreen}{rgb}{0, 0.5, 0}
\definecolor{darkblue}{rgb}{0,0,0.5}

\hypersetup{
	colorlinks=true,
	linkcolor=darkred,
	citecolor=darkgreen,
	urlcolor=darkblue
}

\usepackage[capitalise,noabbrev,nameinlink]{cleveref}

\theoremstyle{plain}
\newtheorem{theorem}{Theorem}
\newtheorem{claim}[theorem]{Claim}
\newtheorem{lemma}[theorem]{Lemma}
\newtheorem*{lemma*}{Lemma}

\newtheorem{corollary}[theorem]{Corollary}

\newtheorem{openquestion}{Open Question}
\newtheorem{openproblem}[openquestion]{Open Problem}

\newtheorem{definition}[theorem]{Definition}

\newtheorem{observation}[theorem]{Observation}
\newtheorem{remark}[theorem]{Remark}

\newcommand{\customthmname}{}
\newtheorem*{customthm*}{\customthmname}

\usepackage[backend=biber,style=ext-alphabetic,url=false,maxbibnames=99,maxalphanames=4,minalphanames=3]{biblatex} 
\DeclareFieldFormat
[article,inbook,incollection,inproceedings,misc,patent,thesis,unpublished]
{titlecase:title}{\MakeSentenceCase*{#1}}

\renewbibmacro{in:}{}

\DeclareFieldFormat{eprint:Electronic Colloquium on Computational Complexity}{Electronic Colloquium on Computational Complexity\addcolon\space
	\ifhyperref
	{\href{\thefield{url}}{\nolinkurl{#1}}}
	{\nolinkurl{#1}}}

\usepackage{algorithmicx}
\usepackage{algorithm}
\usepackage[noend]{algpseudocode}
\usepackage{varwidth}

\usepackage{tikz}
\usetikzlibrary{positioning}
\usetikzlibrary{arrows.meta}
\usetikzlibrary{patterns}
\usepackage{wrapfig}

\usepackage{pdflscape}

\usepackage{booktabs}

\usepackage{enumitem}

\newcommand{\bfA}{\mathbf A}

\newcommand{\bfC}{\mathbf C}

\newcommand{\clsACC}{\bfA\bfC\bfC}

\newcommand{\FunctionName}[1]{\textup{\textsc{#1}}}

\newcommand{\ExactlyN}{\FunctionName{ExactlyN}}

\newcommand{\Equality}{\FunctionName{Equality}}

\newcommand{\AP}[1]{#1{\text -}\mathrm{AP}}
\newcommand{\threeAP}{\AP{3}}
\newcommand{\kAP}{\AP{k}}

\newcommand{\PP}[2]{#1{\text -}\mathrm{PP}_{#2}}
\newcommand{\PPcc}[3]{(#1, #2){\text -}\mathrm{PP}^{\mathrm{cc}}_{#3}}
\newcommand{\vecPP}[2]{#1\text{-}\mathrm{vecPP}_{#2}}
\newcommand{\vecPPcc}[3]{(#1, #2)\text{-}\mathrm{vecPP}^{\mathrm{cc}}_{#3}}

\newcommand{\kPP}[1]{\PP{k}{#1}}
\newcommand{\kPPcc}[2][N]{\PPcc{k}{[#1]}{#2}}
\newcommand{\kPPccspecial}[2]{\PPcc{k}{#1}{#2}}
\newcommand{\kvecPP}[1]{\vecPP{k}{#1}}

\newcommand{\kvecPPcc}[2][\qd]{\vecPPcc{k}{#1}{#2}}
\newcommand{\modulo}[1]{\;(\bmod\; #1)}

\newcommand{\base}{\mathrm{base}}

\DeclareRobustCommand{\eulerian}{\genfrac<>{0pt}{}}

\usepackage{todonotes}

\title{An improved protocol for $\ExactlyN$ with more than 3 players}
\author{Lianna Hambardzumyan \thanks{The Hebrew University of Jerusalem. \texttt{lianna.hambardzumyan@mail.huji.ac.il}. Research partially supported by ISF grant 921/22} 
\and Toniann Pitassi \thanks{Columbia University. \texttt{tonipitassi@gmail.com}. Supported by NSF AF:Medium 2212136. } 
\and Suhail Sherif \thanks{LASIGE, Faculdade de Ci{\^e}ncias, Universidade de Lisboa. \texttt{suhail.sherif@gmail.com}. Funded by the European Union (ERC, HOFGA, 101041696). Views and opinions expressed are however those of the author(s) only and do not necessarily reflect those of the European Union or the European Research Council. Neither the European Union nor the granting authority can be held responsible for them.
Also supported by FCT through the LASIGE Research Unit, ref. UIDB/00408/2020 and ref. UIDP/00408/2020. Most of the work was done while the author was at Vector Institute, Toronto.}
\and Morgan Shirley \thanks{ University of Toronto. \texttt{shirley@cs.toronto.edu}. Supported by an NSERC grant.} 
\and Adi Shraibman \thanks{The Academic College of Tel Aviv-Yaffo. \texttt{adish@mta.ac.il}} 
}

\date{}

\begin{document}

\maketitle
\begin{abstract}
The $\ExactlyN$ problem in the number-on-forehead (NOF) communication setting asks $k$ players, each of whom can see every input but their own, if the $k$ input numbers add up to $N$. Introduced by Chandra, Furst and Lipton in 1983, $\ExactlyN$ is important for its role in understanding the strength of randomness in communication complexity with many players. It is also tightly connected to the field of combinatorics: its $k$-party NOF communication complexity is related to the size of the largest corner-free subset in $[N]^{k-1}$. 

In 2021, Linial and Shraibman gave more efficient protocols for $\ExactlyN$ for 3 players. As an immediate consequence, this also gave a new construction of larger corner-free subsets in $[N]^2$. Later that year Green gave a further refinement to their argument. These results represent the first improvements to the highest-order term for $k=3$  since the famous work of Behrend in 1946. In this paper we give a corresponding improvement to the highest-order term for all $k>3$, the first since Rankin in 1961. That is, we give a more efficient protocol for $\ExactlyN$ as well as larger corner-free sets in higher dimensions.

Nearly all previous results in this line of research approached the problem from the combinatorics perspective, implicitly resulting in non-constructive protocols for $\ExactlyN$. Approaching the problem from the communication complexity point of view and constructing explicit protocols for $\ExactlyN$ was key to the improvements in the $k=3$ setting. As a further contribution we provide explicit protocols for $\ExactlyN$ for any number of players which serves as a base for our improvement.
\end{abstract}
\pagebreak
\floatname{algorithm}{Protocol} \crefname{algorithm}{protocol}{protocols}
\Crefname{algorithm}{Protocol}{Protocols}

\section{Introduction}

In this paper we continue a recent line of work that seeks to apply ideas from \emph{communication complexity} to the field of \emph{additive combinatorics}. Specifically, we study the following problems:

\begin{itemize}
    \item[(i)] {\bf $k$-AP Problem:} What is the maximum size of a subset of $[N]$ that contains no (nontrivial) $k$\emph{-term arithmetic progression} ($\kAP$ for short) -- a sequence
    $x,x+\delta, x+2\delta,\ldots,x+(k-1)\delta$ for some $\delta \neq 0$?
    \item[(ii)] {\bf Corners Problem:} What is the maximum size of a subset of $[N]^k$ that contains no $k$\emph{-dimensional corner} -- a set of  $k+1$ points of the form: \[{(x_1,x_2,\ldots,x_k)},{(x_1 +\delta,x_2,\ldots,x_k)},{(x_1,x_2+\delta,\ldots,x_k)}, \ldots, {(x_1,x_2,\ldots,x_k+\delta)}\] for some $\delta \neq 0$?
\end{itemize}

Our paper is inspired by a growing body of equivalences that have been discovered between problems in additive combinatorics and communication complexity. We build on recent work that exploits these equivalences to gain new perspectives on the two main problems above. 

\begin{itemize}
\item[(i)] The $k$-AP Problem is equivalent to
the deterministic \emph{number-in-hand} (NIH) $k$-player communication complexity of the following promise version of $\Equality$: Each of the $k$ players is given an input $x_i \in [N]$ and they want to decide if their inputs are all equal under the promise that they form a $k$-term arithmetic progression.

\item[(ii)] The Corners Problem is
equivalent to the $(k+1)$-player \emph{number-on-forehead} (NOF) communication complexity of $\ExactlyN$: There are $k+1$ inputs, $x_1,\ldots,x_{k+1} \in [N]$, where Player $i$ sees all inputs except for $x_i$, and they want to decide whether or not the sum of their inputs is equal to $N$.
\end{itemize}

The main contribution of this paper is a new protocol for the $\ExactlyN$ problem that is more efficient than previously-known protocols when there are more than three players. This in turn gives a new method for constructing corner-free subsets of $[N]^k$ which improves on previous constructions for all $k > 2$.

\subsection{Background}

Computational complexity and additive/extremal combinatorics have enjoyed a rich interaction in the last fifty years.
On one side, extremal combinatorics has been critical for proving complexity lower bounds. For example, the Sunflower Lemma underlies Razborov's superpolynomial monotone circuit lower bound \cite{razborov-clique} as well as recent query-to-communication lifting theorems \cite{LovettMMPZ22}, and 
Ramsey's Theorem underlies many complexity lower bounds 
 \cite{pudlak-ramsey}. On the other
side, tools from complexity theory have been used
to resolve problems in additive/extremal combinatorics.
For example,  the recent breakthrough on the Sunflower conjecture \cite{alweiss2020improved} uses ideas behind the Switching Lemma, and the resolution of the Kakeya conjecture \cite{dvir-kakeya} and the Cap-Set Conjecture \cite{clp:capset,eg:capset} use the polynomial method from circuit complexity.
Moreover, some of the main achievements in theoretical computer science -- advances in error correcting codes, the PCP theorem, and pseudorandomness/extractors -- have rich and deep connections
with additive combinatorics~\cite{Lovett-additive}.

In this paper we continue in this tradition by studying two fundamental problems that are well-studied from both the lenses of additive combinatorics and communication complexity. We give a brief discussion  of their importance and motivations from these respective fields.

\paragraph{Additive combinatorics.}
A basic question in number theory and additive combinatorics is understanding the existence of additive structure in the natural numbers, and understanding how much of this structure is algebraic or combinatorial in nature. A remarkable early theorem from 1927 due to Van der Waerden states that for every $r$ and $k$, there exists $N$ such that any $r$-coloring of the numbers in $[N]$ contains a monochromatic $k$-term arithmetic progression. Later it was famously shown that in fact any dense enough subset of the natural numbers contains an arbitrarily large arithmetic progression. Subsequently, many generalizations and quantative versions have received a lot of attention in Ramsey theory, with Szemeredi's Theorem and the Multidimensional Szemeredi's Theorem proving that the density of $\kAP$ free sets and corner-free sets must be sub-constant. This has led to a lot of interest both in improving the density upper bounds and in finding large $\kAP$ and corner-free sets. We refer the reader to the excellent books by Tao and Vu~\cite{tao2006additive} and by Zhao~\cite{yufei2023graph} for a comprehensive treatment.

\paragraph{Communication complexity.} The additive combinatorics problems we study here, viewed through the lens of communication complexity, are essentially questions about derandomization.  The $k$-AP problem, reformulated as a communication problem, is a restriction of the $\Equality$ function, which in the NIH model is easy for randomized protocols but maximally hard for deterministic protocols. The restricted version here asks how  the deterministic complexity changes under the assumption that the inputs have an additive structure.

$\ExactlyN$ (the Corner's problem) has also been studied for the purpose of showing a separation between randomized and deterministic NOF communication complexity. Although a strong non-constructive separation is known even for $k = n^{\epsilon}$ many players~\cite{beameSeparatingDeterministicRandomized2010}, it was only recently that the first constructive separation was shown~\cite{kelley-lovett-meka}, and even then it has only been proven for $k = 3$ players.

 Even though a constructive separation is now known, $\ExactlyN$ continues to be of central importance in this line of research. This is because $\ExactlyN$ is a ``graph function'', and the strong non-constructive separation mentioned above~\cite{beameSeparatingDeterministicRandomized2010} is witnessed by most graph functions. The separating function of~\cite{kelley-lovett-meka} is surprisingly not a graph function, and their lower bound technique is not known to apply to $\ExactlyN$. New techniques developed for lower bounding the complexity of $\ExactlyN$ would then be promising to provide lower bounds when $k>3$. This would be of much interest since NOF lower bounds when $k > \log n$ would imply breakthrough $\clsACC$ circuit lower bounds \cite{bt-acc,yaoACCThresholdCircuits1990}. On the other hand, it is entirely possible that there are efficient protocols for $\ExactlyN$ that are waiting to be discovered.

\subsection{Previous bounds}
The current state-of-the-art reveals a significant difference in our
understanding of the $k$-AP problem and the Corners problem.

\paragraph{The $\kAP$ Problem.} A construction by Behrend from 1946 yields a $3$-AP-free subset of $[N]$ of size at least
${N \cdot 2^{-2 \sqrt{2} \sqrt{\log N} + o(\sqrt{\log N})}}$~\cite{behrendSetsIntegersWhich1946}.\footnote{All logarithms in this paper are base 2.}
The recent breakthrough result of Kelley and Meka \cite{kelley-meka} shows that the exponent is tight to within polynomial factors for $k=3$.

Behrend's result was extended to all $k>3$ by Rankin \cite{rankinSetsIntegersContaining1961} who obtained the following subset size lower bound: 
\[N \cdot 2^{-t 2^{(t-1)/2} \cdot {(\log N)}^{1/t} + o((\log N)^{1/t}) },\]
for $t = \lceil \log k \rceil$. Note that this matches Behrend's result (and, indeed, is the same construction) when $k=3$. The best size upper bound for $k=4$ is $N \cdot 1/(\log N)^{\Omega(1)}$~\cite{greenNewBoundsSzemeredi2017} and for $k>4$ is $N \cdot 1/(\log \log N)^{\eta}$, where $\eta = 2^{-2^{k+9}}$~\cite{gowersNewProofSzemeredi2001}.

\paragraph{The Corners Problem.} Until fairly recently, the best corner-free set construction was via a direct reduction to the $\kAP$ Problem. Ajtai and Szemer{\'e}di first gave this reduction for $k=3$~\cite{ajtaiSetsLatticePoints1974}; their proof easily generalizes to $k>3$. The reduction is very clean and yields the same \emph{density} lower bounds for the $(k-1)$-dimensional Corners Problem as for the $\kAP$ Problem -- if $[N]$ has a $\kAP$-free subset of size $N \cdot \delta$, then $[N]^{k-1}$ has a corner-free subset of size $N^{k-1} \cdot \delta$. In particular, the estimates of Behrend and Rankin can be directly applied to the Corners Problem.

Unlike the $\kAP$ problem, where for $k=3$ relatively tight bounds are known, there is a large gap between upper and lower bounds for the 2-dimensional Corners Problem. The best known upper bound is $N^2 \cdot 1/(\log \log N)^c$ for some constant $c$ by Shkredov~\cite{shkredovGeneralizationSzemerediTheorem2006}. For $k \geq 3$ the best upper bound is just $N^{k} \cdot \omega(1)$~\cite{gowers2007hypergraph}.

Recent works have improved the Ajtai-Szemer{\'e}di reduction, yielding better lower bounds for the Corners Problem, by examining the problem through a communication complexity lens. 

\paragraph{Communication complexity and improved bounds for the Corners Problem.} In 1983, Chandra, Furst, and Lipton defined the NOF model of communication and showed the equivalence between the $k$-party NOF complexity of $\ExactlyN$ and the $(k-1)$-dimensional Corners Problem~\cite{MultipartyPchandrarotocols1983}. 

Specifically,  the minimal cost of protocols for these problems is (up to a constant factor) the logarithm of the optimal solution for the closely-related \emph{coloring} version of the additive combinatorics problems in question:

\begin{itemize}
    \item[(i)] {\bf $k$-AP Problem (Coloring Version):} What is the minimum number of colors to color $[N]$ such that each color class is free of $k$-APs?
    \item[(ii)] {\bf Corners Problem (Coloring Version):} What is the minimum number of colors to color $[N]^k$ such that each color class is free of $k$-dimensional corners?
\end{itemize}

By a standard probabilistic tiling argument the coloring and subset size formulations of these problems are roughly equivalent. A $\kAP$-free subset with size $N/\delta$ implies a $\kAP$-free coloring with $\delta \cdot O(\log N)$ colors, and a similar connection holds for the corners problem. Therefore, a \emph{lower} bound on the size of a $\kAP$-free subset (resp.\ corner-free subset) is the same as an \emph{upper} bound on the $\kAP$-free coloring number (resp.\ corner-free coloring number) and consequently on the NIH complexity of $\Equality$ with a $\kAP$ promise (resp.\ the NOF complexity of $\ExactlyN$).

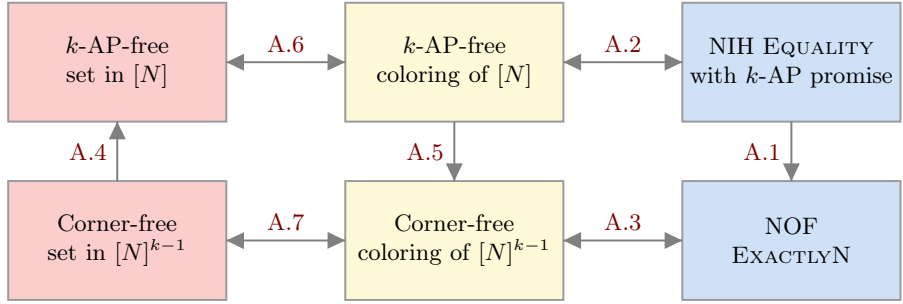
\begin{figure}[t]
    \centering

\tikzset{every picture/.style={line width=0.75pt}} 

\begin{tikzpicture}[x=0.75pt,y=0.75pt,yscale=-1,xscale=1]
\draw  [color={rgb, 255:red, 155; green, 155; blue, 155 }  ,draw opacity=1 ][fill={rgb, 255:red, 246; green, 68; blue, 68 }, fill opacity=0.27 ][line width=0.75]  (70,100) rectangle (180,160) node[midway,align=center,color=black,opacity=1,font=\footnotesize] {$\kAP$-free\\set in $[N]$};
\draw  [color={rgb, 255:red, 155; green, 155; blue, 155 }  ,draw opacity=1 ][fill={rgb, 255:red, 255; green, 236; blue, 112 }  ,fill opacity=0.27 ] (240,100) rectangle (350,160) node[midway,align=center,color=black,opacity=1,font=\footnotesize] {$\kAP$-free\\coloring of $[N]$};
\draw  [color={rgb, 255:red, 155; green, 155; blue, 155 }  ,draw opacity=1 ][fill={rgb, 255:red, 74; green, 144; blue, 226 }  ,fill opacity=0.27 ] (410,100) rectangle (520,160) node[midway,align=center,color=black,opacity=1,font=\footnotesize] {NIH $\Equality$\\with $\kAP$ promise};
\draw  [color={rgb, 255:red, 155; green, 155; blue, 155 }  ,draw opacity=1 ][fill={rgb, 255:red, 246; green, 68; blue, 68 }  ,fill opacity=0.27 ][line width=0.75]  (70,190) rectangle (180,250) node[midway,align=center,color=black,opacity=1,font=\footnotesize] {Corner-free\\set in $[N]^{k-1}$};
\draw  [color={rgb, 255:red, 155; green, 155; blue, 155 }  ,draw opacity=1 ][fill={rgb, 255:red, 255; green, 236; blue, 112 }  ,fill opacity=0.27 ] (240,190) rectangle (350,250) node[midway,align=center,color=black,opacity=1,font=\footnotesize] {Corner-free\\coloring of $[N]^{k-1}$};
\draw  [color={rgb, 255:red, 155; green, 155; blue, 155 }  ,draw opacity=1 ][fill={rgb, 255:red, 74; green, 144; blue, 226 }  ,fill opacity=0.27 ][line width=0.75]  (410,190) rectangle (520,250) node[midway,align=center,color=black,opacity=1,font=\footnotesize] {NOF\\$\ExactlyN$};

\draw [color={rgb, 255:red, 128; green, 128; blue, 128 }  ,draw opacity=1 ]   (125,190) -- node[left, font=\footnotesize] {\ref{sec:ap_set_to_corner_set}} (125,160);
\draw [shift={(125,160)}, rotate = 90] [fill={rgb, 255:red, 128; green, 128; blue, 128 }  ,fill opacity=1 ][line width=0.08]  [draw opacity=0] (8.93,-4.29) -- (0,0) -- (8.93,4.29) -- cycle;
\draw [color={rgb, 255:red, 128; green, 128; blue, 128 }  ,draw opacity=1 ]   (465,190) -- node[left, font=\footnotesize] {\ref{sec:nih_to_ap}} (465,160) ;
\draw [shift={(465,190)}, rotate = 270] [fill={rgb, 255:red, 128; green, 128; blue, 128 }  ,fill opacity=1 ][line width=0.08]  [draw opacity=0] (8.93,-4.29) -- (0,0) -- (8.93,4.29) -- cycle    ;
\draw [color={rgb, 255:red, 128; green, 128; blue, 128 }  ,draw opacity=1 ]   (180,130) -- node[above, font=\footnotesize] {\ref{sec:ap_coloring_to_ap_set}} (240,130) ;
\draw [shift={(240,130)}, rotate = 180.25] [fill={rgb, 255:red, 128; green, 128; blue, 128 }  ,fill opacity=1 ][line width=0.08]  [draw opacity=0] (8.93,-4.29) -- (0,0) -- (8.93,4.29) -- cycle    ;
\draw [shift={(180,130)}, rotate = 0.25] [fill={rgb, 255:red, 128; green, 128; blue, 128 }  ,fill opacity=1 ][line width=0.08]  [draw opacity=0] (8.93,-4.29) -- (0,0) -- (8.93,4.29) -- cycle    ;
\draw [color={rgb, 255:red, 128; green, 128; blue, 128 }  ,draw opacity=1 ]   (180,220) -- node[above, font=\footnotesize] {\ref{sec:corner_color_to_corner_set}} (240,220) ;
\draw [shift={(240,220)}, rotate = 180.5] [fill={rgb, 255:red, 128; green, 128; blue, 128 }  ,fill opacity=1 ][line width=0.08]  [draw opacity=0] (8.93,-4.29) -- (0,0) -- (8.93,4.29) -- cycle    ;
\draw [shift={(180,220)}, rotate = 0.5] [fill={rgb, 255:red, 128; green, 128; blue, 128 }  ,fill opacity=1 ][line width=0.08]  [draw opacity=0] (8.93,-4.29) -- (0,0) -- (8.93,4.29) -- cycle    ;
\draw [color={rgb, 255:red, 128; green, 128; blue, 128 }  ,draw opacity=1 ]   (410,220) -- node[above, font=\footnotesize] {\ref{sec:nof_equiv_corner_coloring}} (350,220) ;
\draw [shift={(350,220)}, rotate = 360] [fill={rgb, 255:red, 128; green, 128; blue, 128 }  ,fill opacity=1 ][line width=0.08]  [draw opacity=0] (8.93,-4.29) -- (0,0) -- (8.93,4.29) -- cycle    ;
\draw [shift={(410,220)}, rotate = 180] [fill={rgb, 255:red, 128; green, 128; blue, 128 }  ,fill opacity=1 ][line width=0.08]  [draw opacity=0] (8.93,-4.29) -- (0,0) -- (8.93,4.29) -- cycle    ;
\draw [color={rgb, 255:red, 128; green, 128; blue, 128 }  ,draw opacity=1 ]   (410,130) -- node[above, font=\footnotesize] {\ref{sec:nih_to_coloring}} (350,130) ;
\draw [shift={ (350,130)}, rotate = 360] [fill={rgb, 255:red, 128; green, 128; blue, 128 }  ,fill opacity=1 ][line width=0.08]  [draw opacity=0] (8.93,-4.29) -- (0,0) -- (8.93,4.29) -- cycle    ;
\draw [shift={(410,130)}, rotate = 180] [fill={rgb, 255:red, 128; green, 128; blue, 128 }  ,fill opacity=1 ][line width=0.08]  [draw opacity=0] (8.93,-4.29) -- (0,0) -- (8.93,4.29) -- cycle    ;
\draw [color={rgb, 255:red, 128; green, 128; blue, 128 }  ,draw opacity=1 ]   (295,190) -- node[left, font=\footnotesize] {\ref{sec:ap_coloring_corner_coloring}} (295,160) ;
\draw [shift={(295,190)}, rotate = 270] [fill={rgb, 255:red, 128; green, 128; blue, 128 }  ,fill opacity=1 ][line width=0.08]  [draw opacity=0] (8.93,-4.29) -- (0,0) -- (8.93,4.29) -- cycle    ;

\end{tikzpicture}
\caption{
The figure shows how the additive combinatorics problems are related to each other and to their communication complexity equivalents. For problems $A$ and $B$, $A \to B$ denotes $c(B) = O(c(A))$, where $c(\cdot)$ measures the problem's complexity in our context. 
} \label{fig:cc-combinatorics}

\end{figure} 
\cref{fig:cc-combinatorics} summarizes the relationships between the problems in additive combinatorics and their communication complexity reformulations. For convenience, we include the proofs of these equivalences in \cref{sec:reductions}. 

The Chandra-Furst-Lipton equivalence, combined with the Ajtai-Szemer\'edi reduction to the $\kAP$ Problem, shows that the NOF communication complexity of $\ExactlyN$ for $k=3$ is at most $2\sqrt{2} \sqrt{\log N} + o(\sqrt{\log N})$ by Behrend's construction, and for $k>3$ is at most $t 2^{(t-1)/2} (\log N)^{1/t} + o\left((\log N)^{1/t}\right)$ for $t = \lceil \log k \rceil$ by Rankin's construction.

The protocols yielded by the above equivalence are \emph{non-explicit}: we have an upper bound on their complexity but the underlying algorithms are non-constructive. This lack of explicitness comes from two places. First, the AP-free subsets of Behrend and Rankin are chosen using a generalized pigeonhole argument. Second, converting the subset size lower bounds into coloring upper bounds requires a probabilistic tiling argument. The problem for us is that such non-explicit protocols are difficult to analyze and therefore difficult to improve. Linial, Pitassi, and Shraibman remedied this situation by giving an explicit protocol for $\ExactlyN$ when $k=3$~\cite{linialCommunicationComplexityHighDimensional2019}.

Recently, Linial and Shraibman gave the first protocol that improves the highest-order term for the $\ExactlyN$ problem for $k=3$ since Behrend's original proof from 1946, yielding also an improved subset size lower bound for the 2-dimensional Corners Problem. Specifically, the constant of $2 \sqrt{2} \approx 2.828$ is improved to $2 \sqrt{\log e} \approx 2.402$~\cite{linialImprovedProtocolExactlyN2021}. This protocol was found by closely examining the explicit protocol of Linial, Pitassi, and Shraibman. Linial and Shraibman's result was further improved by Green, who lowered the constant to $2 \sqrt{2 \log(4/3)} \approx 1.822$~\cite{greenLowerBoundsCornerfree2021}. 

\subsection{Main result}
In this paper, we begin by giving an explicit protocol for $\ExactlyN$ with cost that matches the construction of Rankin. Then we identify an optimization of this protocol which we exploit to give the first improvement in the highest-order term for every constant $k$:

\begin{theorem}\label{thm:main_result}
The number-on-forehead communication complexity of $\ExactlyN$ with $k$ players is at most  
$$\left(1-\frac{c_k}{t}\right)t2^{(t-1)/2}(\log N)^{1/t} + o((\log N)^{1/t}),$$
where $t=\lceil \log k \rceil$ and $c_k$ is a constant depending on $k$.
\end{theorem}
\begin{corollary}
The improved protocol from \cref{thm:main_result} yields a corner-free subset of $[N]^{k-1}$ of size:
\[N^{k-1} \cdot 2^{-\left(1-\frac{c_{k}}{t}\right)t2^{(t-1)/2}{(\log N)}^{1/t} + o((\log N)^{1/t})}\]
for $t=\lceil\log k \rceil$.\footnote{To get a corner-free set of $[N]^{k-1}$ we need to consider the $\ExactlyN$ problem where the inputs of $k$ player are from $[(k-1)N]$  and add up to $(k-1)N$ (see~\cref{sec:nof_equiv_corner_coloring}). This results in extra terms depending on $k$ which can be pushed to the lower order term.}
\end{corollary}

This is the first improvement in the higher-order term since Rankin's 1961 construction. (Rankin's construction gives the above bound but where $c_k=0$ for all $k$.)
Similar to the recent breakthrough due to Linial and Shraibman \cite{linialImprovedProtocolExactlyN2021} and Green \cite{greenLowerBoundsCornerfree2021}, our protocol achieves a constant factor improvement, and for $k=3$ we match Green's bound.  

\begin{remark} \label{remark:annuli}
    In this paper, we are focused on improving the highest-order term in the bounds. However, we would like to highlight the work that has been done on improving the lower-order term as well. Elkin improved the lower-order term in Behrend's construction~\cite{elkinImprovedConstructionProgressionfree2011} (see also the note of Green and Wolf~\cite{greenNoteElkinImprovement2010}) and Elkin's ideas were translated to Rankin's construction by O'Bryant~\cite{obryantSetsIntegersThat2008}. Hunter~\cite{hunterCornerfreeSetsTorus2022} used similar techniques to improve the lower-order term of Green's construction. We leave applying these ideas to our new construction as an open problem (see \cref{sec:open_problems}).
\end{remark}

\paragraph{Outline of Paper.} In \cref{sec:history}, we give a history of the $\ExactlyN$ problem, including an outline of previous results based on Behrend and Rankin, which we hope helps the reader gain an intuition for the remainder of the paper. At the end of \cref{sec:history}, we give an overview of our improved upper bound. 
In \cref{sec:protocol_from_rankin} we give an explicit protocol for $\ExactlyN$ for all $k$, building heavily on Rankin's construction. 
In \cref{sec:nof_protocol_for_exactlyN}, we give our improved protocol, proving \cref{thm:main_result}.
We conclude with some open problems in \cref{sec:open_problems}. \cref{sec:reductions} contains the proofs of the equivalences given in \cref{fig:cc-combinatorics} and \cref{sec:euler} explains how to calculate the value of the constant $c_k$ in \cref{thm:main_result}. 
 \section{Overview of protocols for NOF \texorpdfstring{$\ExactlyN$}{ExactlyN}}\label{sec:history} 
The history of the $\ExactlyN$ problem begins with the paper of Chandra, Furst, and Lipton that defines the NOF communication model~\cite{MultipartyPchandrarotocols1983}. 
By establishing a connection to the Corners Problem they obtained a non-constructive protocol for $\ExactlyN$ with cost $O(\sqrt{\log N})$, beating the cost of the trivial protocol. As mentioned in the introduction, an essential step in this protocol is a reduction to a \emph{promise} instance of the $\Equality$ function in the NIH model. The reduction is outlined in detail in~\cref{sec:nih_to_ap} and is summarized below.

\paragraph{NOF \texorpdfstring{$\ExactlyN$}{ExactlyN} to \texorpdfstring{$\kAP$}{k-AP}-free coloring.} First, the players each perform a reduction that yields the values $X_1, \ldots, X_k$ where $X_i$ is known only to Player~$i$. These values are promised to be a $\kAP$ and are equal if and only if the original instance of $\ExactlyN$ evaluates to $1$. Then Player 1 announces the color of $X_1$ according to some agreed-upon $\kAP$-free coloring of $[kN]$: this is a coloring where no monochromatic subset of $[kN]$ has elements which form a non-trivial $\kAP$. Each other player then sends a single bit for whether or not the color of $X_i$ agrees with the color that Player 1 sent. They all agree if and only if $X_1, \ldots, X_k$ are all equal, as the $\kAP$ promise implies that the colors can not be the same unless $X_1, \ldots, X_k$ are a trivial $\kAP$.

As discussed in the introduction the $\ExactlyN$ problem and the Corners problem in combinatorics are equivalent. Thus the Chandra-Furst-Lipton reduction can be seen as a reduction from the Corners problem to the problem of finding $\kAP$-free colorings (see  \cref{sec:ap_coloring_corner_coloring}). This latter reduction was already known before Chandra-Furst-Lipton  connected these concepts to communication complexity (see~\cite{ajtaiSetsLatticePoints1974} for the case of $k=3$). 

\paragraph{\texorpdfstring{$\kAP$}{k-AP}-free coloring to \texorpdfstring{$\kAP$}{k-AP}-free set.}\label{par:coloring_to_set} The reduction step of the protocol described above is conceptually simple. The technical part is finding a $\kAP$-free coloring of of $[N]$ where the number of colors is minimized.\footnote{The range of integers is $[N]$, instead of $[kN]$ as in the protocol; if we assume that $k$ is a constant this will not affect much.} This number can be estimated by the density version of the coloring problem: find the  largest $\kAP$-free subset of $[N]$.

 By a standard argument these problems are equivalent: a $\kAP$-free subset with size $N/\delta$ implies a $\kAP$-free coloring with $\delta \cdot O(\log N)$ colors (for details, see \cref{sec:ap_coloring_to_ap_set}) and therefore gives a protocol with cost $\log \delta + O(\log \log N)$.  Every known subset construction requires $\delta$ to be superlogarithmic in $N$, in which case the $O(\log \log N)$ term is negligible. Indeed, for $k=3$ we know that superlogarithmic $\delta$ is necessary~\cite{kelley-meka}.
 
 In the rest of the paper we will switch freely between the coloring problems and their subset-size versions.

\subsection{\texorpdfstring{$\ExactlyN$}{ExactlyN} with 3 players} \label{subsec:overview-three-players}
By the Chandra-Furst-Lipton reduction outlined above, a construction of a $\threeAP$-free subset of $[N]$ will result in a protocol for 3-player $\ExactlyN$. Here we summarize the construction of a $\threeAP$-free subset due to Behrend~\cite{behrendSetsIntegersWhich1946}. All of the best known constructions of $\kAP$-free sets are essentially modifications of Behrend's basic framework.

Following prior work of Salem and Spencer~\cite{salemSetsIntegersWhich1942}, Behrend represents numbers in $[N]$ as vectors in $[q]^d$, where $q$ and $d$ are parameters to be chosen later subject to $q^d \geq N$. These vectors are the base-$q$ representations of numbers in $[N]$:
$$\base_{q, d}(x) := (x_0, \ldots, x_{d-1}) \in [q]^d \mbox{ such that } x = \sum_{i = 0}^{d - 1} q^i x_i.$$

The idea behind Behrend's construction is that no three vectors in $[q]^d$ that form a line can lie on the same sphere. Suppose we had the following property: if three numbers $x, y, z \in [N]$ form a $\threeAP$, then their corresponding vectors $\base_{q, d}(x), \base_{q, d}(y), \base_{q, d}(z)$ are in a line. Then one could choose the preimage of any sphere in $[q]^d$ to be the $\threeAP$-free set -- no three distinct vectors in this sphere could be in a line, and so no three distinct numbers in the preimage could form a non-trivial $\threeAP$.

Unfortunately, a $\threeAP$ in $[N]$ does not always correspond to a line in $[q]^d$. This is because of the possibility of carries: as a simple example, 9, 12, and 15 are a $\threeAP$ but the vectors $(0, 9), (1, 2), (1, 5) \in [10]^2$ are not in a line. The strategy that Behrend takes is to avoid carries by limiting the $\ell_\infty$ norm of the vectors. Under this restriction there can never be any carries and so the desired property holds!

We now outline the complete argument. For $\ell \in [dq^2]$, define $A_\ell$ as the set of $x \in [N]$ such that each coordinate of $\base_{q, d}(x)$ has value less than $q/2$ and $\|\base_{q, d}(x)\|_2^2 = \ell$. Then $A_\ell$ is $\threeAP$-free. Furthermore, $\sum_\ell |A_\ell| = (q/2)^d$, so, by pigeonhole principle, for some value of $\ell$ we must have $|A_\ell| \geq \frac{(q/2)^d}{dq^2}$. To optimize this expression we set $d = \sqrt{2 \log N}$ and $q = N^{1/d}$. This gives us a $\threeAP$-free set of size at least ${N \cdot 2^{-2 \sqrt{2} \sqrt{\log N} + o(\sqrt{\log N})}}$, which via the Chandra-Furst-Lipton reduction results in an $\ExactlyN$ protocol of cost $2 \sqrt{2} \sqrt{\log N} + o(\sqrt{\log N})$. 

\paragraph{Explicit and improved protocols.}
From Behrend's construction, the Chandra-Furst-Lipton reduction 
shows the existence of better-than-trivial protocols for $\ExactlyN$. We would like to give a more \emph{explicit} protocol, as an analysis of the details of the protocol may lead to new insights to construct better protocols (and corner-free sets). This motivaton led to the better 3-player $\ExactlyN$ protocols of Linial, Pitassi, and Shraibman~\cite{linialCommunicationComplexityHighDimensional2019}, which was followed by Linial and Shraibman~\cite{linialImprovedProtocolExactlyN2021} and Green~\cite{greenLowerBoundsCornerfree2021}.\footnote{Green's improvement is not phrased as a communication protocol, but was developed after further analyzing the Linial-Shraibman protocol.}

The first explicit protocol of \cite{linialCommunicationComplexityHighDimensional2019} had the general idea to go through the Chandra-Furst-Lipton reduction, yielding values $X_1, X_2, X_3$; player 1 will communicate the (squared) length of $\base_{q, d}(X_1)$, and the other players should agree with this length if and only if $X_1 = X_2 = X_3$. Of course, this runs up against the same carry problem as in Behrend's construction, and here we do not have the liberty of excluding some vectors, as we want this protocol to work for every possible input. Linial, Pitassi, and Shraibman remedy this by having the players explicitly communicate information about the carry. Importantly, their protocol relies on the fact that each input can be seen by two players. The cost of the Linial-Pitassi-Shraibman protocol matches the cost of the non-constructive protocol from Chandra-Furst-Lipton.

Linial and Shraibman~\cite{linialImprovedProtocolExactlyN2021} observed that with the knowledge of two of the inputs, certain carries in the base-$q$ sum of the inputs are more likely than others.  In particular, the entropy of the carry (conditioned on the information shared by certain players) is less than $d$. Linial and Shraibman give a  small-cost protocol that only works for the inputs that have the most likely carry. Then, they show how to translate the inputs on which their protocol does not work to those that do. This process uses communication equal to the entropy of the carry. The total cost of this $\ExactlyN$ protocol is $2 \sqrt{\log e} \sqrt{\log N} + o(\sqrt{\log N})$. Subsequent work of Green refined the argument of Linial and Shraibman and yields a protocol with cost $2 \sqrt{2\log \frac{4}{3}} \sqrt{\log N} + o(\sqrt{\log N})$~\cite{greenLowerBoundsCornerfree2021}.

\subsection{\texorpdfstring{$\ExactlyN$}{ExactlyN} with more than 3 players}

Ideas from Behrend's construction can be used to build a larger $\kAP$-free set for $k > 3$. Rankin was the first to give such a construction~\cite{rankinSetsIntegersContaining1961}; see also the independent rediscovery of this result by {\L}aba and Lacey for a different presentation of the proof~\cite{labaSetsIntegersNot2001}.

The key to Rankin's construction is that the line on which the three vectors fall in the intuition to Behrend's construction can be replaced with a higher-degree object as long as the number of vectors is sufficiently high. This motivates the definition of \emph{polynomial progressions}.
\begin{definition}\label{def:kpp}
A tuple of integers $(x_1,\dots,x_k) \in \mathbb{Z}^k$ is a $k$-term degree-$m$ polynomial progression (denoted $\kPP{m}$) if there is a degree-$m$ polynomial $p$ such that $\forall i \in [k],\,x_i = p(i)$.
\end{definition}
\begin{definition}\label{def:kvecpp}
A tuple of vectors over the integers $(v_1,\dots,v_k) \in (\mathbb{Z}^d)^k$ is a $k$-term degree-$m$ vector polynomial progression (denoted $\kvecPP{m}$) if there are degree-$m$ polynomials $p_j$ for each dimension $j \in [d]$ such that $\forall i \in [k],\,v_i = (p_1(i),...,p_d(i))$.

This definition can be rephrased to say that these are tuples of vectors where each dimension is a $\kPP{m}$.
\end{definition}
Note that a $\kPP{1}$ is just a $\kAP$ and a $\kvecPP{1}$ is just a sequence of vectors equally spaced on a line. Now we can update our intuition of Behrend's construction to include higher-degree progressions, and make an additional observation that will allow us to exploit this fact.
\begin{itemize}
    \item Behrend relies on the fact that no three distinct vectors on a line in $\mathbb{R}^d$ can all be on a sphere. This is the special case of a more general fact: no $2m+1$ vectors that form a $\kvecPP{m}$ are all on a sphere.
    \item If a sequence of vectors form a $\kvecPP{m}$, their squared lengths form a $\kPP{2m}$.
\end{itemize}

We begin by using the first observation to find a $\kPP{m}$-free set where $m$ is a power of two and satisfies $2m+1 \geq k$. This is done in a similar fashion to Behrend's construction: using a pigeonhole argument, choose the preimage of a large set of vectors with the same length. 

Now we can use this $\kPP{m}$-free set (call this set $S$) to find a larger $\kPP{m/2}$-free set. For each $s \in S$, add all of the vectors of squared length $s$ to our new set. The fact that this is $\kPP{m/2}$-free follows from the second observation above: any $\kPP{m/2}$ here would correspond to a $\kPP{m}$ in $S$. We repeat this process, halving the degree at each step, until we have a set with no $\kPP{1}$, i.e.\ a $\kAP$-free set.

In this outline we have omitted many details. In particular, just as in Behrend's construction vectors must be excluded from consideration based on their $\ell_\infty$ norm to avoid carries. Indeed, this exclusion is much stronger than in Behrend's construction: at the step for degree $m$, the set of allowed vectors has density exponentially small in $m$. Fortunately this deficiency is more than compensated for by the fact that vectors of many lengths, instead of simply one length, are included in the sets after the first step. 

If we set the parameters correctly at every step, Rankin's construction gives a $\kAP$-free set of size at least $N \cdot 2^{-t 2^{(t-1)/2} (\log N)^{1/t} + o((\log N)^{1/t})}$ where $t = \lceil \log k \rceil$. For $k = 3$ and $k = 4$, this matches Behrend's construction, which is expected as the construction is exactly the same. For $k \geq 5$, though, there is an improvement in the exponent of the $\log N$ term. Consequently, the cost of the protocol for $\ExactlyN$ from this construction is $t2^{(t-1)/2}{(\log N)}^{1/t} + o((\log N)^{1/t})$ for $t= \lceil \log k \rceil$.

 \subsection{Our results} \label{subsec:our_results}
Our first result gives an explicit protocol for $\ExactlyN$ with any number of players which matches the cost of the non-explicit protocol implied by Rankin. Our second result is an improved protocol for $\ExactlyN$ for more than 3 players that takes advantage of information shared by the players to improve the reduction to the NIH promise $\Equality$ problem. 
  
\paragraph{Sketch of explicit protocol} (For full details, see \Cref{sec:protocol_from_rankin}.)  The idea of this protocol is depicted in~\cref{fig:reductions}.
 As in the previous protocols, the players first locally perform the reduction to NIH $\Equality$ problem with the promise that the new values $X_1, \ldots X_k$ form a $\kAP$. Then each player computes the base-$q$ representation vector of their inputs and the problem reduces to checking vector-$\Equality$ ($\Equality$ over vectors) with the promise that the input vectors form a $\kvecPP{1}$. Next, they compute the squared length of these vectors and reduce to $\Equality$ with  $\kPP{2}$ promise. Although this promise is not as strong as the promise of being a $\kAP$, the reduction is helpful since their new inputs are much smaller than their initial inputs. The players continue by converting their new inputs into base-$q$ representation vectors again, and then computing the lengths of those vectors and so on. Thus, they keep reducing $\Equality$ with  $\kPP{m}$ promise to  vector-$\Equality$ with $\kvecPP{m}$ promise and vector-$\Equality$ with $\kvecPP{m}$ promise to $\Equality$ with  $\kPP{2m}$ promise. When reducing the vector-$\Equality$ to $\Equality$ the degree of polynomial progression in the promise doubles, but the input size decreases in each reduction. When reducing  $\Equality$ to vector-$\Equality$ the degree as well as the input size stays the same, and the input is now a vector polynomial progression which allows us to continue with the reductions.

 This process can repeat at most $\lceil \log k \rceil$ times, as when the degree $m\geq k-1$, the promise $\kPP{m}$ is trivially satisfied. At this point, the players are left to solve the $\Equality$ problem on their current inputs. So one of the players communicates the final length, and all the other players verify whether they have the same length.

 To avoid carries during the process, every time the players reduce $\Equality$ to vector-$\Equality$, they need to make sure that all the obtained vectors are small. If they are not small, one of the players computes and announces a translation which will make her vector small, referred to in this paper as the \emph{shift}. If other players need different shifts, then the vectors are not equal, and we can terminate. Otherwise, all the players shift their vectors by the same amount before computing the lengths of the vectors again. 

\paragraph{Sketch of improved protocol.} (For full details see \Cref{sec:nof_protocol_for_exactlyN}.) Recall that the goal of the players is to figure out whether $\sum_{i \in [k]} x_i = N$. The protocols that arise from previous constructions of corner-free sets involve computing the values $\base_{q, d}(x_i)$, the base-$q$ representations of the players' inputs, thus creating a vector variant of the task in $d$-dimensional space. Unfortunately, just as in the explicit protocol above, there is the possibility of \emph{carries}. Therefore, it is not necessarily the case that $\sum_{i \in [k]} \base_{q, d}(x_i)$ is equal to $\base_{q, d}\left(\sum_{i \in [k]} x_i\right)$. 

Previous protocols~\cite{linialCommunicationComplexityHighDimensional2019,linialImprovedProtocolExactlyN2021} have leveraged the NOF setting to have the players reason about the exact form of the carries. Specifically, these protocols have the players communicate information about the \emph{carry string}: the length-$d$ string representing the carries performed in the summation. We take the same approach. 

Let us rephrase the objective as figuring out whether $\sum_{i \in [k-1]} x_i = N - x_k$. Player $k$ can then look at the base-$q$ representations of the $x_i$s that they see and compute the carries required in the summation on the left-hand side of the expression. They can then convey the carry string to the other players. By adjusting the inputs accordingly, the players can end up with vectors $v_1$ to $v_{k-1}$ that actually do add up to the base-$q$ representation of the left-hand side of the expression as desired. With this strategy each entry of the carry string takes a value between 0 and $k-2$, so $d \log (k-1)$ bits of communication are required. 

We can use the information shared by the players to lower the cost of this even further: we have not yet exploited the fact that each of the first $k-1$ players know $k-2$ of the inputs in the sum. Indeed, in the view of any of the first $k-1$ players there are only two values that each coordinate of the carry string can take, and these values are consecutive. Therefore, if the $k$th player simply communicates the \emph{parity} of each coordinate of the carry string, each other player will have enough information to reconstruct the full carry string. This improves the communication to $d$ bits. 

Note that using $d$ bits to communicate the carry matches the cost of just directly reducing it to an NIH problem and then switching to base-$q$ representations in the NIH model as in the explicit protocol above (see \cref{fig:reductions}); we need one final trick to find an advantage. Let us first consider the case where $k$ is even (so we are adding an odd number of vectors). In this situation it is more likely for the parities of entries in the carry string to take value $0$, where probability is over the uniform distribution on the inputs. The idea is to use a protocol that assumes that the input is ``nice'': one where the parity-of-carry string takes the most likely value of $0$ in every coordinate. If the input is indeed nice, the players simply proceed as if the $k$th player had communicated the all-$0$ string. Otherwise, we use communication to shift the inputs so that they fulfill the assumption.

The cost of this protocol is $d(1 - \Omega(1))$ bits. The reason this is more efficient is that a larger-than-$2^{-d}$ fraction of inputs are nice, and hence (using a set-covering argument) fewer than $2^d$ possible shifts are required.

When $k$ is odd (so we are adding an even number of vectors), the fraction of nice inputs is $2^{-d}$. So the protocol as described above is more efficient only when $k$ is even. This can be rectified by considering the centered base-$q$ representations, where instead of using the digits $0,\dots,q-1$ we use the digits $\lceil -(q-1)/2 \rceil,\dots,\lfloor q/2 \rfloor$. This representation results in a larger-than-$2^{-d}$ fraction of nice inputs both when $k$ is even and when $k$ is odd.

 \section{Explicit NIH protocol for Rankin}\label{sec:protocol_from_rankin}

In this section we give an explicit protocol for the number-in-hand $\Equality$ problem with the promise that the inputs form a $\kAP$ that matches the cost of the non-explicit protocol guaranteed by Rankin's construction. As mentioned in the previous section, the general strategy of our protocol is to convert the $\kAP$ to a higher-degree polynomial progression by converting the integers into vectors, finding the squared length of those vectors (which leaves the parties again with integers), and repeating the process. Converting integers to vectors requires some care, and sidestepping potential problems in this step is the main technical contribution of this section.

Recall the definitions of $\kPP{m}$ and $\kvecPP{m}$ (\cref{def:kpp,def:kvecpp}). We define related communication tasks below. We define the following communication tasks, which are versions of the $\Equality$ problem with the promise that the inputs form either a $\kPP{m}$ or $\kvecPP{m}$.

\begin{definition}
The communication task $\kPPcc{m}$ is defined as follows. 
\begin{itemize}
    \item The input $(x_1,\dots,x_k) \in [N]^k$ is promised to be a $\kPP{m}$.
    \item The output is $1$ if $x_1 = \cdots = x_k$ (referred to as a trivial $\kPP{m}$) and $0$ otherwise.
\end{itemize} 
\end{definition}
\begin{definition}\label{def:kvecppcc}
The communication task $\kvecPPcc{m}$ is defined as follows. 
\begin{itemize}
    \item The input $(v_1,\dots,v_k) \in ([q]^d)^k$ is promised to be a $\kvecPP{m}$.
    \item The output is $1$ if $v_1 = \cdots = v_k$ (referred to as a trivial $\kvecPP{m}$) and $0$ otherwise.
\end{itemize} 
\end{definition}

We make the following observations about these tasks.

\begin{observation} \label{obs:vec-as-and}
    $\kvecPPcc{m}$ is equivalent to $\FunctionName{AND}_d \circ \kPPcc[q]{m}$. That is, $(v_1,\dots,v_k)$ is a valid input for $\kvecPPcc{m}$ if and only if for each $i \in [d]$, $(v_{1,i},\dots,v_{k,i})$ is a valid input to $\kPPcc[q]{m}$. Furthermore, the output on $(v_1,\dots,v_k)$ is $1$ if and only if the output  of $\kPPcc[q]{m}$ on each $(v_{1,i},\dots,v_{k,i})$ is $1$.
\end{observation}

\begin{observation} \label{obs:trivial-when-m-large}
    When the degree $m$ is large enough, the promise in these tasks becomes trivially fulfilled. When $m \geq k-1$, any $(x_1,\dots,x_k) \in [N]^k$ is a valid input to $\kPPcc{m}$. This is because you can find a degree $k-1$ polynomial $p$ such that $p(i) = x_i$ for all $i \in [k]$. Hence for $m \geq k-1$, $\kPPcc{m}$ is equivalent to the \FunctionName{Equality} function. Similarly for $m \geq k-1$, $\kvecPPcc{m}$ is also equivalent to the \FunctionName{Equality} function.
\end{observation}

In this section we show explicit protocols exhibiting the following upper bound for the communication tasks.

\begin{theorem}\label{thm:nihcomplexities}
    Let $m \leq k-1$ and $t = \lceil \log(k/m) \rceil$. Then the number-in-hand communication complexity of computing $\kPPcc{m}$ is at most
    \[ t 2^{(t-1)/2}\sqrt[t]{m^{t-1}\log N} + O(tk^2 \log \log N). \]
    For $m \leq (k-1)/2$, the number-in-hand communication complexity of $\kvecPPcc{m}$ is at most
    \[ (t-1) 2^{(t-2)/2}\sqrt[t-1]{(2m)^{t-2}\log (q^2d)} + O(tk^2 \log \log (q^2d)). \]
\end{theorem}

 As a special case (setting $m = 1$) this yields the desired protocol for NIH $\Equality$ with $\kAP$ promise. See \cref{fig:reductions} for an illustration. The figure also shows where our improvement for NOF $\ExactlyN$ comes into play; this is described in detail in \cref{sec:nof_protocol_for_exactlyN}.

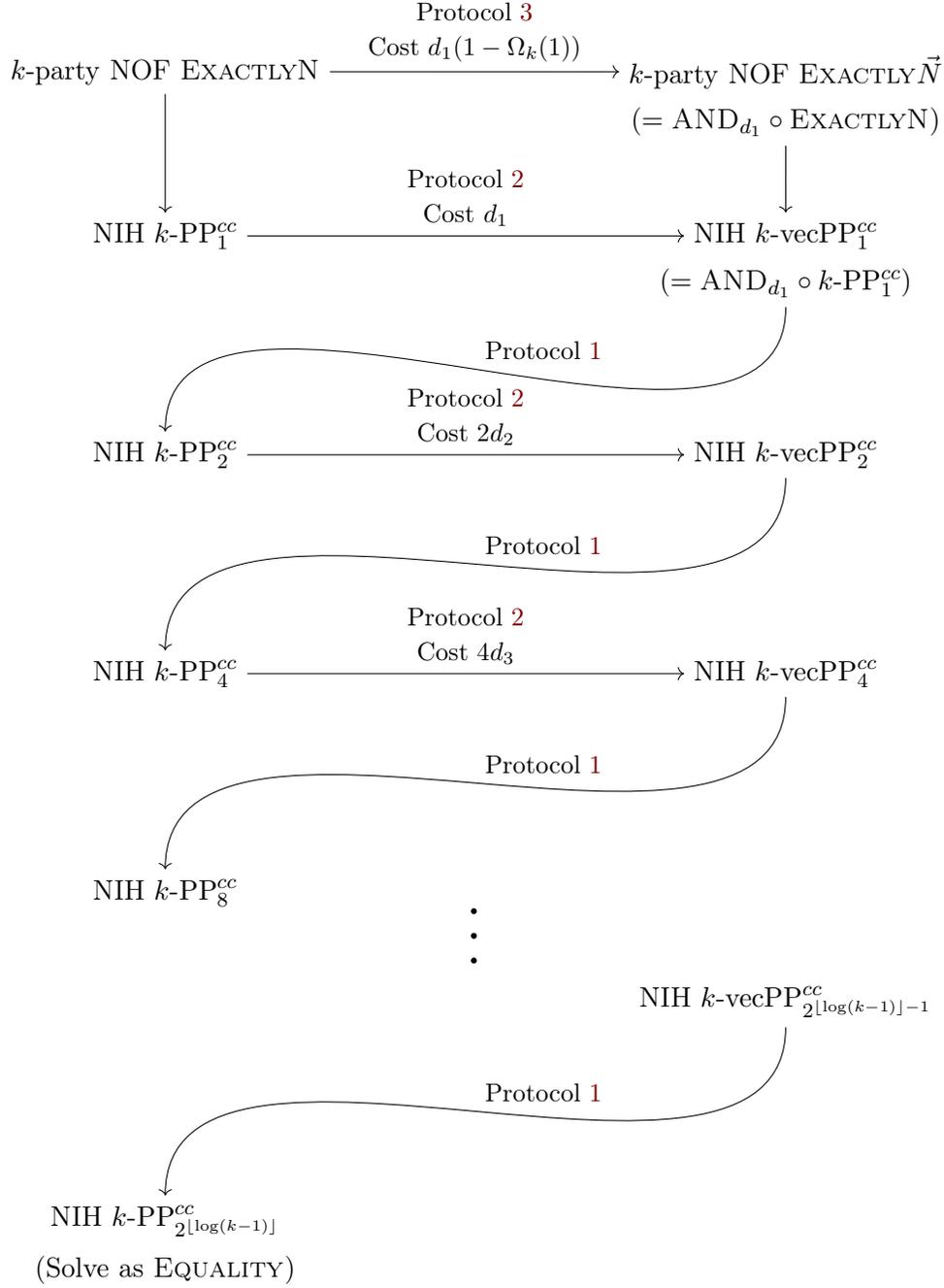
\begin{figure}
\centering
\begin{tikzpicture}
    [xscale=1.7,yscale=1.5,
    old/.style = {}]\node (exac) at (0,11.5) {$k$-party NOF $\ExactlyN$};
    \node (vexac) at (5,11.5) {$k$-party NOF $\FunctionName{Exactly}\vec{N}$};
    \node[below=0pt of vexac] (vexacexplain) {($= \FunctionName{AND}_{d_1} \circ \ExactlyN$)};
    \node (pp1) at (0,10) {NIH $\kPP{1}^{cc}$};
    \node (vpp1) at (5,10) {NIH $\kvecPP{1}^{cc}$};
    \node[below=0pt of vpp1] (vpp1explain) {($= \FunctionName{AND}_{d_1} \circ \kPP{1}^{cc}$)};
    \node (pp2) at (0,8) {NIH $\kPP{2}^{cc}$};
    \node (vpp2) at (5,8) {NIH $\kvecPP{2}^{cc}$};
    \node (pp4) at (0,6) {NIH $\kPP{4}^{cc}$};
    \node (vpp4) at (5,6) {NIH $\kvecPP{4}^{cc}$};
    \node (pp8) at (0,4) {NIH $\kPP{8}^{cc}$};
    \node[scale=2,rotate=90] (dots) at (2.5,3.65) {$\cdots$};
    \node (vppkby2) at (5,3) {NIH $\kvecPP{2^{\lfloor \log(k-1) \rfloor-1}}^{cc}$};
    \node (ppk) at (0,1) {NIH $\kPP{2^{\lfloor \log(k-1) \rfloor}}^{cc}$};
    \node[below=0pt of ppk] {(Solve as $\Equality$)};

    \draw[->,old] (exac)  -- (pp1);
    \draw[->,old] (pp1) -- node[auto,align=center] {\small Protocol~\ref{alg:reductionkpptokvpp}\\\small Cost $d_1$} (vpp1);
    \draw[->,old] (vpp1explain) to[out=270,in=90] node[auto,swap,align=center] {\small Protocol~\ref{alg:reductionkvpptokpp}} (pp2);
    \draw[->,old] (pp2) -- node[auto,align=center] {\small Protocol~\ref{alg:reductionkpptokvpp}\\\small Cost $2d_2$} (vpp2);
    \draw[->,old] (pp4) -- node[auto,align=center] {\small Protocol~\ref{alg:reductionkpptokvpp}\\\small Cost $4d_3$} (vpp4);
    \draw[->,old] (vpp2) to[out=270,in=90] node[auto,swap,align=center] {\small Protocol~\ref{alg:reductionkvpptokpp}} (pp4);
    \draw[->,old] (vpp4) to[out=270,in=90] node[auto,swap,align=center] {\small Protocol~\ref{alg:reductionkvpptokpp}} (pp8);

    \draw[->,old] (vppkby2) to[out=270,in=90] node[auto,align=center,swap] {\small Protocol~\ref{alg:reductionkvpptokpp}} (ppk);

\draw[->] (exac) to node[auto,align=center] {\small Protocol~\ref{alg:reductionexactlyntokvexac}\\\small Cost $d_1(1-\Omega_k(1))$} (vexac);
    \draw[->] (vexacexplain) to (vpp1);
    
\end{tikzpicture} 
\caption{The list of reductions used in protocols for $k$-party NOF \ExactlyN. Reductions that do not mention a cost are 0-cost reductions.} \label{fig:reductions}
\end{figure} 
The proof of \cref{thm:nihcomplexities} is given in~\cref{subsec:proofofnihtheorem}. It uses as subroutines two protocols that we present and analyze below.

\begin{itemize}
    \item \Cref{alg:reductionkvpptokpp} gives us a way to reduce the vector polynomial progression task $\kvecPPcc{m}$ to the integer polynomial progression task $\kPPcc[q^2d]{2m}$ as long as $k > 2m$. Note that we have made the problem harder by moving from degree $m$ to degree $2m$ but we have also decreased the input size from $d \log q$ bits per input to $2 \log q + \log d$ bits per input.
    \item \Cref{alg:reductionkpptokvpp} gives us a way to reduce the integer polynomial progression task $\kPPcc{m}$ to the vector polynomial progression task $\kvecPPcc{m}$. This protocol uses $md$ bits of communication and requires that $q$ is a multiple of $2^m$, $q^d \geq N$ and $k \geq m+2$.
\end{itemize}

\newlength{\reqwidth}
\settowidth{\reqwidth}{\textbf{Promise:}}
\algrenewcommand{\algorithmicrequire}{\makebox[\reqwidth][l]{\textbf{Input:}}}
\algrenewcommand{\algorithmicensure}{\makebox[\reqwidth][l]{\textbf{Output:}}}
\algnewcommand{\algorithmicpromise}{\makebox[\reqwidth][l]{\textbf{Promise:}}}
\algnewcommand\Promise{\item[\algorithmicpromise]}

\begin{algorithm}[ht]
\caption{A reduction from $\kvecPPcc{m}$ to $\kPPcc[q^2d]{2m}$}\label{alg:reductionkvpptokpp}
\begin{algorithmic}[1]
\Require    \begin{varwidth}[t]{0.9\linewidth}
                  $v_1,v_2,\dots,v_k \in [q]^d$ distributed among $k$ players in the NIH model
            \end{varwidth}
\Promise    \begin{varwidth}[t]{0.9\linewidth}
                $v_1,v_2,\dots,v_k$ form a $\kvecPP{m}$ with $k > 2m$
            \end{varwidth}
\Ensure     \begin{varwidth}[t]{0.9\linewidth}
                $x_1,x_2,\dots,x_k \in [q^2d]^k$ distributed among the $k$ players in the NIH model such that $x_1,x_2,\dots,x_k$ form a $\kPP{2m}$, trivial if and only if $(v_1,\dots,v_k)$ is trivial 
            \end{varwidth}
\State For each $i \in [k]$, Player $i$ computes $x_i := \|v_i\|^2$.
\end{algorithmic}
\end{algorithm}
\begin{algorithm}[ht]
\caption{A reduction from $\kPPcc{m}$ to $\kvecPPcc{m}$}\label{alg:reductionkpptokvpp}
\begin{algorithmic}[1]
\Require    \begin{varwidth}[t]{0.9\linewidth}
                  $x_1,x_2,\dots,x_k \in [N]$ distributed among $k$ players in the NIH model \par
                  (Assume $2^m | q, q^d \geq N$)
            \end{varwidth}
\Promise    \begin{varwidth}[t]{0.9\linewidth}
                $x_1,x_2,\dots,x_k$ form a $\kPP{m}$ with $k \geq m+2$
            \end{varwidth}
\Ensure     \begin{varwidth}[t]{0.9\linewidth}
                $v_1,v_2,\dots,v_k \in ([q]^d)^k$ distributed among the $k$ players in the NIH model such that either\par
                (a) $v_1,v_2,\dots,v_k$ form a $\kvecPP{m}$, trivial if and only if $(x_1,\dots,x_k)$ is trivial, or\par
                (b) $x_1,x_2,\dots,x_k$ was a non-trivial $\kPP{m}$ and at least one of the players knows this.
            \end{varwidth}
\State For each $i \in [k]$, Player $i$ computes $w_i \gets \base_{q,d}(x_i)$.
\State $c \gets q/2^m$
\State For each $i \in [k]$, Player $i$ computes two vectors:
\begin{itemize}
    \item $s_i = (\lfloor w_{i,1}/c \rfloor, \dots, \lfloor w_{i,d}/c \rfloor)$ and
    \item $v_i = (w_{i,1} \modulo{c}, \dots, w_{i,d} \modulo{c})$.
\end{itemize}\label{line:computeshift}
\State Player $1$ broadcasts $s_1$.\label{line:player1commshift}
\State For each $i \in [k]$, Player $i$ checks if $s_i = s_1$. If they are not equal, player $i$ notes that the input was a non-trivial $\kPP{m}$.\label{line:playernotes}
\end{algorithmic}
\end{algorithm}

\subsection{Analysis of \texorpdfstring{\Cref{alg:reductionkvpptokpp}}{Protocol 1}}

The input $(v_1,\dots,v_k)$ is promised to be a $\kvecPP{m}$. Let $p_1, \dots,p_d$ be the degree-$m$ polynomials associated with them, in the sense that $v_i = (p_1(i),\dots,p_d(i))$. Define the degree-$2m$ polynomial $p' := \sum_{j \in [d]} p_j^2$. Note that the $x_i$ computed in the protocol is merely $p'(i)$. Hence $(x_1,\dots,x_k)$ is a $\kPP{2m}$. If the original $\kvecPP{m}$ was trivial, then the computed $\kPP{2m}$ is also trivial. On the other hand if any $p_j$ is non-constant, then $p'$ is also non-constant (any monomial of maximal degree among the $p_j$s will get squared and hence not get cancelled in $p'$). Assuming $k > 2m$, the non-constant polynomial $p'$ cannot take the same value on $k$ different points and so the $\kPP{2m}$ is non-trivial.

The cost of this protocol is $0$ since there is no communication during the protocol.

\subsection{Analysis of \texorpdfstring{\Cref{alg:reductionkpptokvpp}}{Protocol 2}}

We start with a useful statement about polynomials. Define the function $L$ as follows:

\[L(a_0,\dots,a_{m+1}) = \sum_{i=0}^{m+1} (-1)^i \binom{m+1}{i} a_i.\]

\begin{claim}[folklore]
    Let $k \geq m+2$. The sequence $(x_1,\dots,x_k)$ forms a $\kPP{m}$ if and only if $$L(x_1,\dots,x_{m+2}) = \dots = L(x_{k-m-1},\dots,x_k) = 0.$$
\end{claim}
\begin{proof}
    This proof follows from properties of a ``difference operator'' $\Delta$ defined on sequences of real numbers. The output of this operator is also a sequence of real numbers containing the differences of the consecutive elements of the input sequence. For instance if $f$ is a sequence $(a_1,a_2,a_3,a_4)$, $\Delta f = (a_2-a_1,a_3-a_2,a_4-a_3)$. In the following observations it will be useful to represent polynomials of degree at most $m$ as linear combinations of the basis polynomials $\binom{x}{0},\dots,\binom{x}{m}$.
    \begin{itemize}
        \item Since $\binom{x+1}{d}-\binom{x}{d} = \binom{x}{d-1}$, if $p$ is a degree-$m$ polynomial then $p(x+1)-p(x)$ would be a degree-$(m-1)$ polynomial. (If $m=0$, it would be the zero polynomial.) Hence if we have a sequence $f = (p(1),p(2),\dots,p(k))$ for some degree-$m$ polynomial $p$, the sequence $\Delta f$ would be of the form $(q(1),q(2),\dots,q(k-1))$ for some degree-$(m-1)$ polynomial $q$.
        \item Conversely, ``inverting $\Delta$'' increases the degree by 1: if a sequence $(q(0),q(1),\dots,q(k-1))$ is obtained by applying $\Delta$ to a sequence $f$, then $f$ must be of the form $a,a+q(0),a+q(0)+q(1),\dots,a+\sum_{i=0}^{k-1}q(i)$ for some real $a$. We know that $\binom{0}{d} + \binom{1}{d} + \dots + \binom{x}{d} = \binom{x+1}{d+1}$. Hence if $q$ is a polynomial of degree $m$ then $\sum_{i=0}^{x} q(x) = p(x)$ for some polynomial $p$ of degree $m+1$. So the sequence $f$ must be of the form $(p(0),p(1),\dots,p(k))$ for some polynomial $p$ of degree $m+1$.
        \item Lastly, we analyze the iterated operator $\Delta^{m+1}$. If $f$ is a sequence $(a_1,a_2,...)$ with $\Delta^{m+1} f = (b_1,b_2,...)$, then by induction $b_i = \sum_{j=0}^{m+1} (-1)^j \binom{m+1}{j} a_{i+j}$. Hence the statement \[L(x_1,\dots,x_{m+2}) = \dots = L(x_{k-m-1},\dots,x_k) = 0\] merely states that $\Delta^{m+1}$ applied to the sequence $(x_1,\dots,x_k)$ results in the all-zero sequence of length $k-m-1$.
    \end{itemize}    

    Now we can complete the proof. Suppose $(x_1,\dots,x_k)$ forms a $\kPP{m}$, which means that $x_i = p(i)$ for some polynomial $p$ of degree $m$. Then applying $\Delta^{m+1}$ to the sequence would give the values of the zero polynomial. Hence $L(x_1,\dots,x_{m+2}) = \dots = L(x_{k-m-1},\dots,x_k) = 0$.

    For the converse, the sequence $(0,0,\dots,0)$ of length $k-m-1$ is a sequence of values of the zero polynomial. Since this sequence is derived from applying $\Delta^{m+1}$ to the sequence $(x_1,\dots,x_k)$, we know that the sequence $(x_1,\dots,x_k)$ must be $(p(1),p(2),\dots,p(k))$ for some polynomial $p$ of degree at most $m$.
\end{proof}

From \cref{obs:vec-as-and} it immediately follows that vectors $(v_1,\dots,v_k)$ form a $\kvecPP{m}$ if and only if $L(v_1,\dots,v_{m+2}) = \dots = L(v_{k-m-1},\dots,v_k) = \vec{0}$.

Now we can analyze the correctness of \Cref{alg:reductionkpptokvpp}. Recall that parameters $q$ and $d$ are set such that $q^d \geq N$ and $q$ is a multiple of $2^m$. Let $S$ be the set of numbers in $[N]$ whose base-$q$ representations only have entries less than $q/2^m$.
\begin{claim}\label{clm:smallentryreduction}
Let $(x_1,\dots,x_k)$ be a $\kPP{m}$ with each $x_i \in S$. Then their base-$q$ representations $(v_1,\dots,v_k)$ form a $\kvecPP{m}$, trivial if and only if the $\kPP{m}$ was trivial.
\end{claim}
\begin{proof}
Define the vector $w$ as $w := L(v_1,\dots,v_{m+2})$. The sum of the positive coefficients in the map $L$ is $\sum_{i \in [m+1],\,i \text{ even}} \binom{m+1}{i} = 2^m$, so each entry in $w$ is less than $q/2^m \cdot 2^m = q$. Similarly we can see that each entry is larger than $-q$. Rearranging the summations in the definition of $w$, we obtain
\[
\begin{aligned}
\sum_{j \in [d]} w_j q^{j-1} &= \sum_{j \in [d]} L(v_{1,j},\dots,v_{m+2,j}) q^{j-1}\\
&= L\left(\sum_{j \in [d]} v_{1,j}q^{j-1},\dots,\sum_{j \in [d]} v_{m+2,j}q^{j-1}\right)\\
&= L(x_1,\dots,x_{m+2}) = 0.
\end{aligned}
\]
The first non-zero entry of $w$, say $w_i$, must be a multiple of $q$, otherwise $\sum w_j q^{j-1} \mod q^i \neq 0$. Since each entry of $w$ is larger than $-q$ and smaller than $q$, $w$ must be equal to $\vec{0}$. The same argument works to show that $L(v_2,\dots,v_{m+3}) = \dots = L(v_{k-m-1},\dots,v_k) = \vec{0}$. So we can conclude that $(v_1,\dots,v_k)$ form a $\kvecPP{m}$. Since the operation of taking the base-$n$ representation is a bijection, $x_1 = \cdots = x_k$ if and only if $v_1 = \dots = v_k$.
\end{proof}

Clearly if in line~\ref{line:playernotes} a player notes that $s_i \neq s_1$, that player's input is different from the input of Player 1, and so the $\kPP{m}$ held by the players must have been non-trivial. We now prove that if no player has $s_i \neq s_1$, then the vectors they compute at the end form a $\kvecPP{m}$. Note that the $v_i$ computed in line~\ref{line:computeshift} can equivalently be written as $v_i := w_i - cs_i$.  Since we are now analyzing the case when the locally-computed $s_i$s are all equal, the vector $v_i$ can be written as $v_i = w_i - cs_1$. Since it lies in $\{0,1,\dots,c-1\}^d$, it is the base-$q$ representation of an integer $T(v_i) : = \sum_j v_{i,j} q^{j-1}$.

Since $T: (a_1,\dots,a_d) \mapsto \sum_j a_j q^{j-1}$
 is a linear transform, $T(v_i) = T(w_i) - T(cs_1)$. We know $T(w_i) = x_i$, so $T(v_i) = x_i - T(cs_1)$. Hence $T(v_1),\dots,T(v_k)$ are just $x_1,\dots,x_k$ shifted by the integer $T(cs_1)$. Hence $T(v_1),\dots,T(v_k)$ also form a $\kPP{m}$. Since every entry of their base-$q$ representation is at most $c-1 < q/2^m$, we can use \Cref{clm:smallentryreduction} to conclude that $v_1,\dots,v_m$ are a $\kvecPP{m}$, trivial if and only if the $x_i$s were. This proves the correctness of the protocol.

The cost of this protocol is $md$ since the only communication that occurs is in Line~\ref{line:player1commshift} where Player 1 broadcasts an element of $\{0,\dots,2^m-1\}^d$.

\subsection{Combining \texorpdfstring{\Cref{alg:reductionkvpptokpp,alg:reductionkpptokvpp}}{Protocols 1 and 2}}\label{subsec:proofofnihtheorem}

Our protocol for $\kPPcc{m}$ uses \Cref{alg:reductionkvpptokpp,alg:reductionkpptokvpp} to repeatedly reduce the problem until it becomes an instance of the form $\kPPcc[N']{m'}$ with $m' \geq k/2$. At this point they can no longer reduce the input size through these reductions, and so they solve this problem as an \FunctionName{Equality} problem: Player $1$ reveals their input and all the other players communicate $0$ if their input differs or if at any point in the reductions via Protocol~\ref{alg:reductionkpptokvpp} they noted that the input was a non-trivial $\kPP{}$~(see line \ref{line:playernotes}). They communicate $1$ otherwise. The output of the protocol is $1$ if all the players communicate $1$. The correctness of this protocol is easy to verify. The cost of the protocol depends on the parameters chosen during the reductions, and we analyze this in the proof.

\begin{proof}[Proof of \cref{thm:nihcomplexities}]
We prove the claim by induction on $t = \lceil \log(k/m) \rceil$.

The base case is when $t=1$, corresponding to $k/2 \leq m \leq k-1$. Since $\kPPcc{m}$ is a promise version of $\Equality$ on $\log N$ bits it can be solved by player $1$ broadcasting their input and the other players using $1$ bit each to convey whether their inputs match that of player $1$. This protocol requires $\log N + k$ bits and works for all $m$.

For the inductive step, let $\lceil \log(k/m) \rceil = i+1$. Since $i+1$ is at least $2$, we have $k > 2m$. This means we can use \Cref{alg:reductionkpptokvpp} to reduce it to $\kvecPPcc{m}$ and then \Cref{alg:reductionkvpptokpp} to reduce that to $\kPPcc[q^2d]{2m}$. Since $\lceil \log(k/2m) \rceil = i$, by our induction hypothesis we already have an upper bound on the communication complexity of $\kPPcc[q^2d]{2m}$.

Going via this reduction we get a protocol of cost \[ md + i 2^{(i-1)/2}\sqrt[i]{(2m)^{i-1} \log q^2d} + O(ik^2 \log \log q^2d), \] assuming $q^d \geq N$ and $q$ is a multiple of $2^m$ (this condition is required for us to run \Cref{alg:reductionkpptokvpp} with cost $md$). We can easily find the minimum of a closely related quantity that captures the essence of the minimization task.

\begin{claim}\label{clm:minimizenih}
\[ \min_{q',d' \in \mathbb{R}_+, q'^{d'}=N} md' + i 2^{(i-1)/2}\sqrt[i]{(2m)^{i-1} \log q'^2} = (i+1) 2^{i/2} \sqrt[i+1]{m^i \log N}, \]
achieved when $md' = 2^{(i-1)/2}\sqrt[i]{(2m)^{i-1} 2 \log q'} = 2^{i/2} \sqrt[i+1]{m^i \log N}$
\end{claim}

\begin{proof}
Since $d' (\sqrt[i]{\log q'})^{i} = \log N$, we have \[(md')\left( 2^{(i-1)/2}\sqrt[i]{(2m)^{i-1} 2 \log q'}\right)^{i} = 2^{i(i+1)/2}m^{i}\log N.\] This is the product of $i+1$ terms: one term is $md'$ and the other $i$ terms are $2^{(i-1)/2}\sqrt[i]{(2m)^{i-1} 2 \log q'}$. The quantity we want to minimize is exactly the sum of these terms. This sum is minimized when each of the terms are the same, and hence equal to the $i+1$th root of the product.
\end{proof}

In our actual minimization problem we want to ensure that $q$ is a multiple of $2^m$ and $d$ is a natural number, and we also are minimizing a larger quantity. In the rest of the proof we show that accounting for these only adds to the lower order term. Let $q',d'$ be the optimal values in \Cref{clm:minimizenih}. We can always find a $q \in [q',q'+2^m)$ and $d \in [d',d'+1)$ that satisfy our conditions. Plugging these in to our original minimization task, we get an upper bound of
\[ m(d'+1) + i2^{(i-1)/2}\sqrt[i]{(2m)^{i-1} \log ((q'+2^m)^2(d'+1))} + O(ik^2 \log \log q^2(d'+1)).\]

Using $\sqrt[i]{a+b} \leq \sqrt[i]{a} + \sqrt[i]{b}$ and $\log(a+b) \leq \log a + \log b$ for $a,b \geq 2$, this is in turn upper bounded by 
\begin{align*}
&md' + i2^{(i-1)/2}\sqrt[i]{(2m)^{i-1} \log q'^2}\\
&+ m + i2^{(i-1)/2}\sqrt[i]{(2m)^{i-1}} (\sqrt[i]{2 \log 2^m} + \sqrt[i]{\log d})
+ O(ik^2 \log \log q^2d).
\end{align*}

We know the first two terms add up to $(i+1) 2^{i/2} \sqrt[i+1]{m^i \log N}$. We analyze the other terms using the fact that $2^{i+1} \leq k/m$.
\begin{itemize}
    \item $i 2^{(i-1)/2}\sqrt[i]{(2m)^{i-1} \log d}$: Since we choose a value of $d$ that is at most $k \sqrt[i+1]{\log N}+1$, this term is at most $\log \log N+1$ when $i=1$ and $o(k^2 \log \log N)$ otherwise.
    \item $i 2^{(i-1)/2}\sqrt[i]{(2m)^{i-1} 2 \log 2^m}$: This is just $i 2^{(i-1)/2} 2m$, which is at most $k$.
    \item $i k^2 \log \log q^2d$: This is at most $i k^2 \log \log N$ since $q^2d \ll q'^{d'} = N$.
    \item $m$ is at most $k$.
\end{itemize}

Hence our final bound is
\[ (i+1) 2^{i/2}\sqrt[i+1]{m^i\log N} + O((i+1)k^2 \log \log N). \qedhere \]
\end{proof} \section{Improved NOF protocol for ExactlyN}\label{sec:nof_protocol_for_exactlyN}
In this section we will show how to use information shared by the players to improve the reduction to the NIH promise $\Equality$ problem.

Recall that the goal of the players is to figure out whether $\sum_{i \in [k-1]} x_i = N - x_k$. We will use the high-level ideas described in~\cref{subsec:our_results}. 
 We now formally define the centered base-$q$ representation and carry-related notions, and then present the protocol.

\subsection{Centered base-\texorpdfstring{$q$}{q} representations, carry strings and carry vectors}

For simplicity, assume $q$ is odd. For an integer $x \in \{-(q^d-1)/2,\dots,(q^d-1)/2\}$, the centered base-$q$ representation of $x$ is a vector $\base^{\pm}_{q,d}(x)$ defined as the unique $v \in \{-(q-1)/2,\dots,(q-1)/2\}^d$ such that $x = \sum_{j \in [d]} v_j q^{j-1}$.

When adding together numbers $x_1$ through $x_t$ which have centered base-$q$ representations $v_1$ through $v_t$, we can get the centered base-$q$ representation of the sum by adding $v_1$ through $v_t$ but then modifying the result to take care of the carries. This is captured by the following process. (We require here that $t<q$, and this will be the case whenever we use this.)
\begin{itemize}
    \item Let $w = \sum_{i \in [t]} v_i$. 
    \item Define a \emph{carry string} $s \in \mathbb{Z}^d$ as follows
    \begin{itemize}
        \item $s_1$ is the unique integer such that $w_{1} \in \{s_1q - (q-1)/2,\dots,s_1q + (q-1)/2\}$.
        \item For $j \in \{2,\dots,d\}$, $s_j$ is the unique integer such that $w_{j} + s_{j-1} \in \{s_jq - (q-1)/2,\dots,s_jq + (q-1)/2\}$.
    \end{itemize}
    \item Define a \emph{carry vector} $v_s \in \mathbb{Z}^{d+1}$ as $\sum_{j \in [d]} s_j(e_{j+1} - qe_j)$.
    \item Then $w + v_s = \base^{\pm}_{q,d+1}\left(\sum_{i \in [t]} x_i\right)$. (Here $w$ is viewed as a $(d+1)$-dimensional vector with $w_{d+1}=0$.)
\end{itemize}

The following claim will be useful for communicating the carry to players in the NOF model.

\begin{claim}\label{clm:reconstructcarry}
Let $v_1,\dots,v_t \in \{-(q-1)/2,\dots,(q-1)/2\}^d$ and $s$ be the carry string of $\sum_{i \in [t]} v_i$. Given only $\{s_j \;(\bmod\; 2)\}_{j \in [d]}$ and $v_2,\dots,v_t$, one can reconstruct $s$ entirely.
\end{claim}
\begin{proof}
We prove this by induction. The base case is that we can reconstruct $s_1$, and the inductive step shows that given $s_{j-1}$ and the information provided to us we can reconstruct $s_j$. Let $v_x = \sum_{i \in \{2,\dots,t\}} v_i$. We can compute $v_x$ with the information provided. Although we do not know $v_1$, we know that each entry of $v_1$ lies in $\{-(q-1)/2,\dots,(q-1)/2\}$.

For the base case, let $\alpha$ be the unique integer such that $v_{x,1} \in \{\alpha q - (q-1)/2,\dots,\alpha q + (q-1)/2\}$. If $v_{x,1} = \alpha q$, then with the addition of $v_{1,1}$ it will still remain in this interval and so $s_1 = \alpha$. If $v_{x,1} < \alpha q$, then with the addition of $v_{1,1}$ it will either remain in the same interval or move to the interval corresponding to $\alpha-1$. So $s_1 \in \{\alpha-1,\alpha\}$. Similarly if $v_{x,1} > \alpha q$, we know $s_1 \in \{\alpha,\alpha+1\}$. In any of these cases finding out $s_1 \modulo{2}$ will specify $s_1$ exactly.

The inductive step is similar. Assume we know $s_{j-1}$. By definition $s_j$ is defined by which interval $v_{x,j} + v_{1,j} + s_{j-1}$ lies in. We know the value of $v_{x,j} + s_{j-1}$ and so again $s_j$ depends on where the addition of $v_{1,j}$ can move it. With the same reasoning as before, finding out $s_j \;(\bmod\; 2)$ will specify $s_j$ exactly.
\end{proof}

\subsection{A reduction to a vector variant}

\Cref{alg:reductionexactlyntokvexac} is a reduction from $\ExactlyN$ to a vector variant that we term $\FunctionName{Exactly}\vec{N}$. In this protocol, players have as inputs (in the NOF model) $x_1,\dots,x_k$. Player $k$ then broadcasts a shift so that all the players can compute new inputs $a_1$ to $a_k$ (still in the NOF model) such that $\sum_{i \in [k]} x_i = N \iff \sum_{i \in [k-1]} a_i = a_k$. These new inputs are also designed to have the property that if you take the base-$q$ representations of these inputs (called $w_1,\dots,w_k$ in the protocol), and you look at the carry string obtained by adding $w_1$ through $w_{k-1}$, all of its entries are even. From \Cref{clm:reconstructcarry}, this will allow all of the players to know the exact carry string $w_s$ and for them to shift the vector $w_k$ by it in order to ensure that $\sum_{i \in [k]} x_i = N \iff \sum_{i \in [k-1]} w_i = w_k - w_s$.

This vector variant of $\ExactlyN$ is then used to create a protocol for $\ExactlyN$ in \Cref{sec:protocolforexactlyn}

\begin{algorithm}[ht]
\caption{A reduction from NOF $\ExactlyN$ to NOF $\FunctionName{Exactly}\vec{N}$}\label{alg:reductionexactlyntokvexac}
\begin{algorithmic}[1]
\Require    \begin{varwidth}[t]{0.9\linewidth}
                  $x_1,x_2,\dots,x_k \in [N]$ are distributed among the $k$ players in the NOF model \par
                  $q^d \geq N$
            \end{varwidth}
\Ensure     \begin{varwidth}[t]{0.9\linewidth}
                $v_1,v_2,\dots,v_k \in \{-kq,\dots,kq\}^{d+1}$ are distributed among the $k$ players in the NOF model, with $\sum_{i \in [k]} v_i = \vec{0}$ if and only if $\sum_{i \in [k]} x_i = N$.
            \end{varwidth}
\State Player $k$ broadcasts a $\delta \in \mathbb{Z}^{k-1}$ such that\begin{enumerate}
    \item[(a)] for each $i \in [k-1]$, $x_i + \delta_i \in \{-(q^d-1)/2,\dots,(q^d-1)/2\}$, and \item[(b)] the assertion in Line~\ref{line:correctcarriespm} holds.\end{enumerate}\label{line:playerkcomm}
\State For $i \in [k-1]$, $a_i \gets x_i + \delta_i$, $a_k \gets N - x_k + \sum_{i \in [k-1]} \delta_i$.
\State For $i \in [k-1]$, let $w_i \gets \base^{\pm}_{q,d}(a_i)$ and let $w_k \gets \base^{\pm}_{q,d+1}(a_k)$.
\State Player $k$ computes $s \in \{-kq,\dots,kq\}^{d}$, the carry string of $\sum_{i \in [k-1]} w_i$.
\State \textbf{Assert:} For each $j \in [d]$, $s_j \modulo{2} = 0$. \label{line:correctcarriespm}
\State For each $i \in [k]$, Player $i$ computes $s$ and the carry vector $w_s$.
\State For each $i \in [k-1]$, $v_i := w_i$ and $v_k := - w_k + w_s$.
\end{algorithmic}
\end{algorithm}

\subsubsection{Correctness of the reduction}

Let us first note that Line~\ref{line:playerkcomm} is always achievable. That is, that there is always a $\delta$ that player $k$ can compute such that the assertion in Line~\ref{line:correctcarriespm} holds. One such $\delta$ is $(-x_1,\dots,-x_{k-1})$, which player $k$ can compute. With this $\delta$, each $a_i$ is $0$ for $i \in [k-1]$. The corresponding $w_i$s would also be $0$ vectors and the carry string of $\sum_{i \in [k-1]} w_i$ would also be a string of $0$s. This carry string satisfies the assertion that for each $j \in [d]$, $s_j \modulo{2}=0$.

Now we prove the correctness of the protocol assuming only that the assertion in Line~\ref{line:correctcarriespm} holds.

\begin{itemize}
    \item We start by showing that $(v_1,\dots,v_k)$ are indeed known to the players in the NOF model. The vector $w_i$ depends only on $x_i$ and $\delta_i$, which are known to all players except player $i$. Since the assertion in Line~\ref{line:correctcarriespm} holds, every player knows that each entry of $s$ is even. Along with the fact that every player misses at most one of the summands in $\sum_{i \in [k-1]} w_i$, from Claim~\ref{clm:reconstructcarry} we see that every player does in fact know the string $s$. The carry vector $w_s$ is a function of $s$, and hence they know $w_s$ as well. The vector $v_i$ depends only on $w_i$ and $w_s$, so all the players other than player $i$ can compute $v_i$.
    \item We finish by showing that $\sum_{i \in [k]} v_i = \vec{0}$ if and only if $\sum_{i \in [k]} x_i = N$.
    \begin{align*}
        \sum_{i \in [k]} x_i = N &\iff \sum_{i \in [k-1]} a_i = a_k \tag{definition of $a_i$'s} \\
        &\iff \sum_{i \in [k-1]} w_i + w_s = w_k \tag{definition of $w_i$'s and the carry vector} \\
        &\iff \sum_{i \in [k]} v_i = \vec{0}. \tag{definition of $v_i$'s}
    \end{align*}
\end{itemize}
It is easy to see that for each $i \in [k-1]$ $v_i \in \{-(q-1)/2,\dots,(q-1)/2\}^d$, (which we will be viewing as a $d+1$-dimensional vector with $v_{i,d+1}=0$). Since $v_k$ has a carry vector added to it, with the carries being as large as $(k-1)q$, $v_k \in \{-kq,\dots,kq\}^{d+1}$.

\subsubsection{Cost of the reduction} \label{subsec:cost-of-protocol}

The communication in the protocol is entirely in Line~\ref{line:playerkcomm}. The cost of this line depends on the size of the smallest set $\Delta \subset \mathbb{Z}^{k-1}$ such that for any $x_1,\dots,x_{k-1} \in [N]$ there exists $\delta \in \Delta$ which satisfies the requirements in Line~\ref{line:playerkcomm}. The communication cost is then merely $\lceil \log |\Delta| \rceil$ since Player $k$ only needs to send the index of an element of $\Delta$.

The size of $\Delta$ is related to the size of the set
\begin{align*}
    S := \{&(a_1,\dots,a_{k-1}) \in \{-(q^d-1)/2,\dots,(q^d-1)/2\}^{k-1} \mid\\ &\text{ the carry string of } \sum_{i \in [k-1]} \base^{\pm}_{q,d}(a_i) \text{ has only even entries}\}.
\end{align*}
$\Delta$ is the smallest set of shifts of $S$ that covers $[N]^{k-1}$. We can show the following bounds on $|\Delta|$. \[ N^{k-1}/|S| \leq |\Delta| \leq ((2q^d)^{k-1}/|S|) \cdot k\log N. \]

The lower bound on $|\Delta|$ is straightforward. For the upper bound we use the probabilistic method. Choose shifts $\delta^{(1)},\dots,\delta^{(t)}$ uniformly at random from $\{-N-(q^d-1)/2,\dots,(q^d-1)/2\}^{k-1}$. For any $\overline{x} = (x_1,\dots,x_{k-1})$, there are exactly $|S|$ different shifts that would land $\overline{x}$ in $S$. Hence the probability that a uniformly random shift is good for $\overline{x}$ is $|S|/(q^d+N)^{k-1} \geq |S|/(2q^d)^{k-1}$. The probability that none of the $t$ shifts are good for $\overline{x}$ is at most $(1-|S|/(2q^d)^{k-1})^t$. Setting $t = ((2q^d)^{k-1}/|S|) \cdot k \log N$, this probability is at most $e^{-k \log N} \leq 1/N^k$. Hence by a union bound over all $N^{k-1}$ possible values of $\overline{x}$, there is a positive probability that (and hence there exists a set of $t$ shifts such that) each $\overline{x}$ has a shift that is good for it.

The cost of the protocol is hence at most $k-1 + \log(q^{d(k-1)}/|S|) + \log k + \log \log N + 1$.

So how large is $S$? Note that the integers from $-(q^d-1)/2$ to $(q^d-1)/2$ have centered base-$q$ representations ranging over all vectors in $\{-(q-1)/2,\dots,(q-1)/2\}^d$. Hence
\begin{align*}
    \frac{|S|}{q^{d(k-1)}} &= \Pr_{x_1,\dots,x_{k-1} \in \{-(q^d-1)/2, \dots, (q^d-1)/2\}} [(x_1,\dots,x_{k-1}) \in S]\\
    &= \Pr_{v_1,\dots,v_{k-1} \in \{-(q-1)/2,\dots,(q-1)/2\}^d}[\text{the carry string of }\sum_{i \in [k-1]} v_i \text{ has only even entries}].
\end{align*}

We now use the following claim which we prove in \Cref{sec:euler}.

\begin{claim}\label{claim:realprobability}
Let $r_1,\dots,r_{k-1}$ be real numbers uniformly sampled from $[-1/2,1/2)$.
\[ \Pr_{r_1,\dots,r_{k-1}} \left[\sum_{i \in [k-1]} r_i \modulo{2} \in [-1/2,1/2)\right] = \frac{1}{2} + \frac{E_{k-1}}{2(k-1)!}, \]
where $E_n$ is the $n$th Euler zigzag number.\footnote{See \Cref{sec:euler} or entry \href{https://oeis.org/A000111}{A000111} in The On-Line Encyclopedia of Integer Sequences (starts at $E_0$) for more details.}
\end{claim}

Observe that the above quantity represents the limiting behaviour, as $q \to \infty$, of a specific entry of the carry string being even. The rest of the proof will show that the probability that a specific entry (say, the $i$th entry) of the carry string is even is within an additive $3k/2q$ of the probability in \Cref{claim:realprobability}, regardless of what we fix the entries of $v_1$ to $v_{k-1}$ to be outside of their $i$th entries.

\begin{itemize}
    \item The probability that $s_1$ is even is the probability that $k-1$ random numbers $a_1,\dots,a_{k-1}$ chosen from $\{-(q-1)/2,\dots,(q-1)/2\}$ add up to give an even carry. Note that the carry is even if and only if the sum modulo $2q$ lies in $\{-(q-1)/2,\dots,(q-1)/2\}$. We approximate this by a probability arising from the following real-valued experiment. Take $k-1$ real numbers $r_1,\dots,r_{k-1}$ from the interval $[-1/2,1/2)$. Find the probability that their sum modulo $2$ lies in $[-1/2,1/2)$. The two processes are related as follows.
    
    Let the set $B = \{-(q-1)/2,\dots,(3q-1)/2\}$ represent the set of integers modulo $2q$. Divide $[-1/2,3/2)$ into $2q$ intervals of size $1/2q$ each. Let $i_1,\dots,i_{k-1}$ be the index of the intervals that $r_1,\dots,r_{k-1}$ lie in. Each $i$ is a uniformly random number from $1$ to $q$, and so $a_j$ is distributed as the $i_j$th element of $B$. Let $i_s$ be the interval that the sum $\sum_j r_j \modulo{2}$ lies in. Then $\sum_j a_j$ modulo $2q$ lies within the $i_s$ through $i_{s+k-2}$th elements of $B$.
    
    So either we have $\sum_j r_j \modulo{1} \in [1/2 - k/2q,1/2)$, or else it must be the case that $\sum_j r_j \modulo{2} \in [-1/2,1/2) \iff \sum_j a_j \modulo{2q} \in \{-(q-1)/2,\dots,(q-1)/2\}$. Hence the difference in probabilities of the experiments is at most $\Pr[\sum_j r_j \modulo{1} \in [1/2 - k/2q,1/2)]$. This is $k/2q$, since the addition of a uniformly random number between $[0,1]$ to any random variable makes its distribution modulo 1 the uniform distribution.
    \item For other coordinates of the carry string another complication arises. Since the sum in a coordinate is the sum of $k-1$ random numbers plus the carry from the previous coordinate, that adds another change in the experiment. However, the carry from the previous coordinate is always within $\{-k+1,\dots,k-1\}$ so it adds an uncertainty of $\pm k/2q$ to the sum in the real-valued experiment. Hence we can use the same real-valued experiment, except this time we bound the difference in probabilities as $\Pr[\sum_j r_j \modulo{1} \in [1/2 - k/q,1/2 + k/2q)] = 3k/2q$.
\end{itemize}

Hence the probability that all entries of the carry string are even is at least $(1/2 + E_{k-1}/2(k-1)! - 3k/2q)^d$. The cost of the protocol is at most
\[ d \log\left(\frac{1}{1/2 + E_{k-1}/2(k-1)! -3k/2q}\right) + k + \log k + \log \log N. \]

Since $k/q \ll 1$ and $\frac{d}{dt}\log\left(\frac{1}{1/2 + t}\right) = -\frac{2}{\ln 2} > -3$ at $t=0$, this quantity is at most $$d\log\left(\frac{1}{1/2 + E_{k-1}/2(k-1)!}\right) + d \cdot \frac{9k}{2q} + O(k + \log \log N),$$ with $9dk/2q$ being $o(1)$ if $d \leq \log N/\log \log N$. In our usage we will have $d \leq \sqrt{\log N}$.

To simplify this expression, define
\begin{equation}\label{eqn:ck}
c_k \triangleq 1- \log\left(\frac{1}{1/2 + E_{k-1}/2(k-1)!}\right).
\end{equation}
As $k$ grows, $c_k \to \frac{2}{\ln 2} \left(\frac{2}{\pi}\right)^k$. \Cref{alg:reductionexactlyntokvexac} uses $(1-c_k)d + O(k + \log \log N)$ bits of communication.

\subsection{Putting everything together}\label{sec:protocolforexactlyn}

Our protocol starts by running \Cref{alg:reductionexactlyntokvexac} with parameters $q,d$ such that $q^d \geq N$. The players end up with vectors $v_1,\dots,v_k$, each in $\{-kq,\dots,kq\}^{d+1}$, (in the NOF setting) and they want to know whether $\sum_{i \in [k]} v_i = \vec{0}$. Note that this sum is equal to $\vec{0}$ if and only if for each $j \in [d+1]$, $\sum_{i \in [k]} v_{i,j} = 0$. Each of these is an instance of $\ExactlyN$ with the inputs coming from $\{-kq,\dots,kq\}$.

At this point, they can make a cost-$0$ reduction to $\kvecPPcc[\{-k^3q,\dots,k^3q\}^{d+1}]{1}$ in the NIH setting. This is because each instance of $\ExactlyN$ has a cost-$0$ reduction to $\kPPccspecial{\{-k^3q,\dots,k^3q\}}{1}$ (as described in \Cref{sec:nih_to_ap}) and because $\kvecPPcc[\{-k^3q,\dots,k^3q\}^{d+1}]{1}$ is equivalent to $\FunctionName{AND}_{d+1} \circ \kPPccspecial{\{-k^3q,\dots,k^3q\}}{1}$ (see Observation~\ref{obs:vec-as-and}).
One should note here that the reduction in \Cref{sec:nih_to_ap} works even when the input is allowed to include negative numbers. This is also true of \Cref{alg:reductionkvpptokpp}, which is the first step in the NIH protocol for $\kvecPPcc{1}$ and which outputs a nonnegative $\kPP{2}$.

We can now use the NIH protocol for $\kvecPPcc[\{-k^3q,\dots,k^3q\}^{d+1}]{1}$ (\cref{thm:nihcomplexities}) to complete the protocol. Let $t = \lceil \log k \rceil$. The cost of the NIH protocol is 
$(t-1) 2^{(t-2)/2}\sqrt[t-1]{2^{t-2}\log (k^6q^2(d+1))} + O(tk^2 \log \log (k^6q^2(d+1)))$.

The total cost of the protocol is then \[ (1-c_k)d + k + \log k + \log \log N + (t-1) 2^{(t-2)/2}\sqrt[t-1]{2^{t-2}\log (k^6q^2(d+1))} + O(tk^2 \log \log (k^6q^2(d+1))). \]
As done in the proof of \cref{thm:nihcomplexities} we can optimize the values of $d$ and $q$ and end up with a complexity of
\begin{align*}
&\,t 2^{(t-1)/2} \sqrt[t]{(1-c_k) \log N} + O(tk^2 \log \log N) \\\leq &\left(1-\frac{c_k}{t}\right)t2^{(t-1)/2}\sqrt[t]{\log N} + O(tk^2 \log \log N).
\end{align*} \section{Open problems}\label{sec:open_problems}

In this paper we give the first explicit protocol for $\ExactlyN$ that matches the performance of Rankin's construction. We then use the details of this explicit protocol to find an improvement that relies on knowledge shared by the parties.

However, this improvement itself relies on an existential argument: there is a probabilistic argument in \cref{subsec:cost-of-protocol}. Therefore our final improved protocol has a non-constructive part.

\begin{openproblem}
	Give a completely explicit protocol that matches the performance of the NOF protocol from~\cref{thm:main_result}.
\end{openproblem} 

The constructions of Behrend and Rankin use a pigeonhole argument over spheres in some vector space. As mentioned in \cref{remark:annuli}, there is a line of work that improves the lower-order terms of these constructions~\cite{elkinImprovedConstructionProgressionfree2011,greenNoteElkinImprovement2010,obryantSetsIntegersThat2008,hunterCornerfreeSetsTorus2022}. The general strategy is to replace the spheres with thin \emph{annuli}. We have not attempted to use annuli in our construction, but it seems to us that this might lead to an improvement in lower order terms in our case too.

\begin{openproblem}
	Improve the lower-order terms of our corner-free set construction by replacing spheres with annuli.
\end{openproblem}

Beyond this, any further improvements in upper or lower bounds for any of the problems discussed in this paper would be important advances on their own terms. We wish to highlight a few directions here that are of particular interest to us.

Our protocol exploits the shared information between the players in the NOF setting. As the number of parties increases the amount of shared information also increases. One might think that this would lead to a corresponding increase in the magnitude of the improvement in the NOF setting over the protocol described in \cref{sec:protocol_from_rankin}, which makes no use of the shared information. However, this is not what we see: the factor of $\left(1 - c_k/t\right)$ from \cref{thm:main_result} actually grows as $k$ increases. 

    \begin{openproblem}
    Give a corner-free set construction whose advantage over Rankin's construction improves as $k$ grows.
\end{openproblem}

The structure of Rankin's protocol seems to necessitate a lack of smoothness in the parameters of the construction. Namely, the best-known $\kAP$-free set construction when $k$ is not of the form $2^t + 1$ (for an integer $t$) is to round down to the nearest such value and proceed with the corresponding construction. Is it possible to obtain a bound that depends on $\log k$ instead of $\lceil \log k \rceil$? This would be exciting as it would require a different argument than the degree-doubling method used by Rankin.

\begin{openproblem}
	Give a $\kAP$-free set construction that improves for each increase of the value $k$.
\end{openproblem}

Finally, an important open problem is to improve the large gap between the upper and lower bounds on the size of corner-free sets, where progress has been stuck for more than 15 years. 
We feel that it may be possible to  substantially improve the NOF communication complexity of $\ExactlyN$, by further exploiting the shared information in the NOF model.  
On the other hand, if substantial improvements are not possible for $\ExactlyN$, strong lower bounds for $\ExactlyN$ would give a  breakthrough separation of deterministic from randomized NOF protocols for an explicit and well-studied function. As mentioned in the introduction, the recent breakthrough result of Kelley and Meka proved an upper bound for $\threeAP$-free sets~\cite{kelley-meka}, nearly matching Behrend's construction.
 
However, corners appear to be a much more complicated combinatorial object, and upper bounds on corner-free sets have historically lagged behind those for $\threeAP$-free sets.
Thus narrowing this gap is an important problem in additive combinatorics as well.

\begin{openproblem}
	Narrow the gap between the best known upper and lower bounds on the NOF complexity of $\ExactlyN$.
\end{openproblem}

 \appendix
\section{Relations between combinatorial and communication problems}\label{sec:reductions}

In this appendix we use the following shorthand notation.

\begin{itemize}
    \item $r_k(N)$ is the maximum size of a subset of $[N]$ that does not contain a $\kAP$.
    \item $c_k(N)$ is the minimum number of colors needed to color $[N]$ such that no $\kAP$ is monochromatic.
    \item $r^{\angle}_{k}(N)$ is the maximum size of a subset of $[N]^k$ that does not contain a $k$-dimensional corner.
    \item $c^{\angle}_{k}(N)$ is the minimum number of colors needed to color $[N]^k$ such that no $k$-dimensional corner is monochromatic.
\end{itemize}

\subsection{Reduction of \texorpdfstring{$\ExactlyN$}{ExactlyN} to NIH \texorpdfstring{$\Equality$}{Equality} with AP promise}\label{sec:nih_to_ap} 

Let $x_1, \ldots, x_k$ be the inputs for $\ExactlyN$ with $k$ players. Each player, based on the other players' inputs, can calculate the value that their input must take in order for $x_1 + \ldots + x_k = N$ to be true; namely
\[x'_i = N - \sum_{\substack{j \in [k]\\j \not= i}} x_j.\]
Each of these guesses differs from the actual input by the same amount: $ \Delta: = x'_i - x_i = T - \sum_{j \in [k]} x_j$, for all $i$.
Next each player attempts to compute the value $X = \sum_{j \in [k]} j x_j$ by replacing their input value (which they do not know) with the guess input calculated above. Thus Player $i$ guesses the following value for $X$: 
\[ X_i = i x'_i + \sum_{\substack{j \in [k]\\j \not= i}} j x_j. \]

Observe that for all $i$ we have $X_i = X - i\Delta$ and therefore the values $X_i$ form a $\kAP$.

The $\kAP$ $(X_1, \ldots, X_k)$ is trivial (i.e. all of the elements of the sequence are equal) if and only if $\Delta = 0$, which occurs if and only if $\sum_{i \in [k]} x_i = N$. In this case, $\ExactlyN(x_1, \ldots, x_k) = 1$. 

\subsection{Equivalence between NIH \texorpdfstring{$\Equality$}{Equality} with AP promise and AP-free coloring number}\label{sec:nih_to_coloring}
 We include the equivalence for only $k=3$ to avoid tedious notation in the proof, but the proof can easily be generalized for more than 3 players.
\begin{lemma}
 The number-in-hand communication complexity of $\Equality(x,y,z)$ with the promise that $x,y,z$ is a $\threeAP$ is $\Theta (\log c_3(N))$.
\end{lemma}

\begin{proof}
Given a coloring of $[N]$ with $c_3(N)$ colors such that no $\threeAP$ is monochromatic, here is a communication protocol
using $2 + \log c_3(N)$ bits: the first player writes the color of her number on the board, and the other two write one bit determining whether 
the color of their numbers is equal to it or not.

Given an NIH communication protocol $\Pi$ for $\Equality(x_1, x_2,x_3)$, we define the following coloring: given $w \in  [N]$ its color is the transcript
(i.e., what is written on the board) of $\Pi$ on the input $(w, w, w)$. We claim that this coloring avoids monochromatic $\threeAP$s. Assume (seeking a contradiction) that this is not true, and let $x+y=2z$ be three distinct numbers that
share the same color. Since the color of a number $w$ is the transcript of $\Pi$ on $(w,w,w)$ it follows that the transcripts
of $\Pi$ on $(x, x, x)$, $(y, y, y)$ and $(z, z, z)$ are all equal. But since each player decides what to communicate based on the prior communication and their own input, this same transcript would be generated on the input $(x, y, z)$.
But since the transcript is the same the protocol would output the same answer on the inputs $(x,x,x)$ and $(x,y,z)$ contradicting the protocol's correctness.
\end{proof}

\subsection{Equivalence between NOF \texorpdfstring{$\ExactlyN$}{ExactlyN} and corner-free  coloring number}\label{sec:nof_equiv_corner_coloring}
\begin{theorem}[\cite{MultipartyPchandrarotocols1983}, \cite{rao2020communication}]\label{thm:nof_equiv_corner_coloring} 

Let $c$ be the NOF communication complexity of the $\ExactlyN$ problem with $k$ players. The following holds:
$$\log{c^{\angle}_{k-1}\left(\frac{N}{k-1}\right)} \leq c \leq \log{c^{\angle}_{k-1}(N)} + k-1.$$ 

\end{theorem}
\begin{proof}
For the upper bound, consider the set of points
\begin{align*} 
S =\{&(x_1, \ldots, x_{k-1}), \\ & (N - \sum_{j\neq 1} x_j, \ldots, x_{k-1}) \\&\quad\quad\quad \ldots, \\& (x_1, \ldots x_{i-1}, N - \sum_{j\neq i} x_j, x_{i+1}, \ldots, x_{k-1}), \\&\quad\quad\quad \ldots, \\& (x_1, \ldots , N-\sum_{j\neq k-1} x_j) \}. 
\end{align*}
$S$ is a $(k-1)$-dimensional corner in $[N]^{k-1}$; namely, \[S =  \{(x_1, \ldots, x_{k-1}),
 (x_1 + d, \ldots, x_{k-1}), 
 \ldots, (x_1, \ldots, x_{k-1}+d) \}\] for  $d = N - (x_1 + \ldots + x_k)$. Now assume $[N]^{k-1}$ is colored by $c^{\angle}_{k-1}(N)$ colors avoiding monochromatic corners. Thus, the points in $S$ receive the same color if and only if the corner is trivial, i.e. $d=0$, which in this case implies $x_1 \ldots + x_k = N$ -- exactly what the protocol needs to check. So the protocol checks whether all the points received the same color: the player that has $x_k$ on its forehead announces the color of $(x_1, \ldots, x_{k-1})$ with $\log{c^{\angle}_{k-1}(N)}$ bits. The other $k-1$ players send a bit each indicating whether the unique point that they can compute has the same color. 

For the lower bound, assume there is a protocol solving $\ExactlyN$ with $c$ bits.
 We show a coloring of $$C:= \underbrace{\left[ \frac{N}{k-1}\right] \times \ldots \times \left[ \frac{N}{k-1}\right]}_{k-1}$$ that avoids monochromatic corners. Color $(x_1, \ldots, x_{k-1}) \in C$ by the transcript of the protocol on input $(x_1, \ldots, x_{k-1}, N-(x_1 + \ldots + x_{k-1}))$. The number of colors is at most $2^c$. Seeking a contradiction, assume the set 
 \begin{align*}
 &(x_1, \ldots, x_{k-1}), \\
 &(x_1 + d, \ldots, x_{k-1}), \\
 &\quad\quad\quad\quad\vdots \\
 &(x_1, \ldots, x_{k-1}+d)
 \end{align*}
 forms a monochromatic corner for some $d > 0$. This means all of the following inputs result in the same transcript.
\begin{align*}
 P = \{&(x_1, \ldots, x_{k-1}, N-(x_1 + \ldots + x_{k-1})), \\
 &(x_1 + d, \ldots, N-(x_1 + d + \ldots + x_{k-1})), \\
 &\quad\quad\quad\quad\quad\quad\quad\vdots \\
 &(x_1, \ldots, x_{k-1}+d, N-(x_1 + d + \ldots + x_{k-1}))\}.
 \end{align*}
 However, this implies the input $p=(x_1, \ldots, x_{k-1}, N-(x_1 + d + \ldots + x_{k-1}))$ also results in the same transcript (every player has an input from $P$ that it cannot distinguish from $p$, and hence at no point in the protocol does the transcript for $p$ deviate from the transcript for the points in $P$). This is a contradiction to the correctness of the protocol as for all of the points in $P$, their coordinates sum up to $N$, but the coordinates of $p$ sum up to $N-d$.
\end{proof}

\subsection{AP-free set induces a corner-free set}\label{sec:ap_set_to_corner_set}

\begin{claim}[\cite{ajtaiSetsLatticePoints1974},\cite{yufei2023graph}]\label{clm:ap_set_to_corner_set}
$r^{\angle}_{k-1}(k^2 N) \geq N^{k-2} \cdot r_{k}(N)$.
\end{claim}
\begin{proof}
 Let $A \subset [N]$ be a $\kAP$-free set that has size $r_k(N)$. Define the set $$Q:= \big\{(x_1, \ldots, x_{k-1}) \in [k^2N]^{k-1} \colon x_1 + 2x_2 + \ldots + (k-1)x_{k-1} \in A + (k^2-1)N \big\}.$$
$Q$ is corner-free. Indeed, assume $Q$ contains the corner \[(x_1, \ldots, x_{k-1}),  (x_1+d, \ldots, x_{k-1}), \ldots, (x_1, \ldots, x_{k-1}+d)\] for some $d>0$. Then, the sequence $s, s+d, \ldots, s+(k-1)d$ with \[s = x_1 + 2x_2 + \ldots + (k-1)x_{k-1} - (k^2-1)N\] is a $\kAP$ in $[N]$.

 For each $a \in A$ there are at least $N^{k-2}$ elements in $[k^2N]^{k-1}$ that satisfy $x_1 + 2x_2 + \ldots + (k-1)x_{k-1} = a + (k^2-1)N$: choose any $x_2,\dots,x_{k-1} \in [N]$ and there exists $x_1 \in [k^2N]$ that makes the equation true. 
 Thus, $|Q| \geq N^{k-2} |A|$. Combining this with  $r^{\angle}_{k-1}(k^2N) \geq |Q|$ concludes the proof.
\end{proof}

\subsection{AP-free coloring implies a corner-free coloring}\label{sec:ap_coloring_corner_coloring}
\begin{claim}[\cite{MultipartyPchandrarotocols1983}]
    $c_{k-1}^{\angle}\left(\frac{N}{k^2}\right) \leq c_{k}(N)$.
\end{claim}
\begin{proof}
Color $[N]$ with $c_k(N)$ colors avoiding $k$-term APs.
Define a map $q : \left[\frac{N}{k^2}\right]^{k-1} \to [N]$ as follows:
$$q(x_1, x_2, \ldots, x_{k-1}) = x_1 + 2 x_2 + \ldots + (k-1)x_{k-1}.$$
Then color each $(x_1, x_2, \ldots, x_{k-1})$ by the color of $q(x_1, x_2, \ldots, x_{k-1})$. This coloring avoids monochromatic corners in $\left[\frac{N}{k^2}\right]^{k-1}$.
Indeed, if $S$ is a monochromatic corner in $\left[\frac{N}{k^2}\right]^{k-1}$, then the set $\{q(s) : s \in S\}$  is a monochromatic $k$-term arithmetic progression in $[N]$ (similar to the proof of \cref{clm:ap_set_to_corner_set}).
\end{proof}

\subsection{AP-free coloring number is equivalent to largest AP-free set size}\label{sec:ap_coloring_to_ap_set}
\begin{theorem}[\cite{MultipartyPchandrarotocols1983}]\label{thm:ap_set_to_coloring}

$ \frac{N}{r_{k}(N)} \leq c_{k}(N) \leq O\left(\frac{N\lg N}{r_{k}(N)}\right).$
    
\end{theorem}
\begin{proof}
    The lower bound is by pigeonhole principle: if $[N]$ is colored with $c_{k}(N)$ colors avoiding monochromatic $k$-term arithmetic progressions, then there must be a color class that has size at least $\frac{N}{c_k(N)}$.

    For the upper bound, let $A \subset [N]$ be the set forming a $k$-term AP with size $r_{k}(N)$. The claim below shows that we can find at most $O\left(\frac{N \lg N}{|A|}\right)$ translates of $A$ that cover $[N]$, thus also avoid $k$-term APs.
\end{proof}
\begin{claim}\label{clm:translate_and_cover}
        For a set $A \subset [N]$, there are $t_1, \ldots, t_\ell \in [-N,N]$ such that $\cup_{i=1}^\ell (t_i + A) = [N]$ and $\ell \leq O\left(\frac{N\lg N}{|A|}\right).$
\end{claim}
\begin{proof}
        Choose $t_1, \ldots, t_\ell$ uniformly randomly from the range $[-N,N]$. Fix $x \in [N]$. The probability that $x$ is not covered by some $t_i + A$ is at most $1-\frac{|t_i + A|}{2N}= 1-\frac{|A|}{2N}$. Hence the probability of some point being not covered by all of the translates is
        $$p \leq N \cdot \left(1 - \frac{|A|}{2N}\right)^\ell.$$
        Then, $p < 1$, if $\ell > \frac{2N\lg N}{|A|}$. Thus, by the probabilistic method, there exists a choice of $t_1, \ldots, t_\ell$ that for some $\ell = O\left(\frac{N\lg N}{|A|}\right)$, the corresponding translates of $A$ cover $[N]$.
\end{proof}

\subsection{Corner-free coloring number is equivalent to the largest corner-free set size}\label{sec:corner_color_to_corner_set}
\begin{theorem}
    $ \frac{N^k}{r^{\angle}_{k}(N)} \leq c^{\angle}_{k}(N) \leq O\left(\frac{N^k \lg N}{r^{\angle}_{k}(N)}\right).$
\end{theorem}
\begin{proof}
    The proof is analogous to the proof of~\cref{thm:ap_set_to_coloring} as \cref{clm:translate_and_cover} can be extended to work for $[N]^k$ and $k$-corners.
\end{proof} \DeclareRobustCommand{\eulerian}{\genfrac<>{0pt}{}}
\newcommand{\oddevent}{\mathcal{E}}

\section{Proof of \texorpdfstring{\cref{claim:realprobability}}{the Euler zigzag number claim}}\label{sec:euler}
In this section $i$ is used to represent the imaginary unit. The Euler zigzag numbers count the number of alternating permutations of a given length. The $k$th Euler zigzag number is denoted $E_k$.

Let $r_1, \ldots, r_k$ be real numbers uniformly sampled from $[-1/2, 1/2)$. Let $\oddevent_k$ be the event \[\sum_{j \in [k]} r_j \modulo{2} \in [-1/2,1/2).\] Recall that our goal is to prove the following:
\begingroup
\renewcommand{\customthmname}{\Cref{claim:realprobability}}
\begin{customthm*}[Restated]
    $\Pr[\oddevent_k] = \frac{1}{2} + \frac{E_k}{2k!}.$
\end{customthm*}
\endgroup

The proof follows in a mostly straightforward way from known results about the Euler zigzag numbers and related quantities. The methods used in the case where $k$ is odd are standard in the literature. When $k$ is even some of the steps in the proof may be novel, albeit not too difficult to extrapolate from the odd case. For completeness we give the details of the proof for both cases.

\paragraph{Tangent numbers, secant numbers, and Euler zigzag numbers.} The \emph{tangent numbers} are the coefficients in the Maclaurin series of the tangent function. Only odd tangent numbers have nonzero value. Similarly, the \emph{secant numbers} are the coefficients in the Maclaurin series of the secant function, and are only nonzero for even indices.

\[ \tan \theta = \sum_{k \geq 1} \frac{\theta^{2k-1}}{(2k-1)!} T_{2k-1} \qquad \mbox{and} \qquad \sec \theta = \sum_{k \geq 0} \frac{\theta^{2k}}{(2k)!} S_{2k}. \]

The Euler zigzag numbers are the coefficients in the Maclaurin series of $\tan \theta + \sec \theta$:
\[ E_k = \begin{cases} T_k & k \mbox{ is odd}\\S_k & k \mbox{ is even}\end{cases}. \]

\paragraph{Trigonometric identities.} We remind the reader of the following equations:
\[ \tan \theta = -i \frac{e^{i \theta} - e^{-i \theta}}{e^{i \theta} + e^{-i \theta}} \qquad \mbox{and} \qquad \sec \theta = \frac{2}{e^{i \theta} + e^{-i \theta}}. \]

\paragraph{Eulerian numbers.} The Eulerian number $\eulerian{k}{\ell}$ is the number of permutations of length $k$ with $\ell$ descents: that is, the number of permutations $\sigma : [k] \rightarrow [k]$ with exactly $\ell$ indices $z$ where $\sigma(z) > \sigma(z+1)$. The Eulerian polynomial is defined as $A_k(t) = \sum_{\ell=0}^{k-1} \eulerian{k}{\ell} t^\ell$. The corresponding exponential generating function is $A(t, u) = \sum_{k \geq 0} A_k(t) \frac{u^k}{k!}$. The identity 
\begin{equation} \label{equ:euler-gen-func} A(t, u) = \frac{t-1}{t-e^{u(t-1)}} \end{equation} 
was proved by Euler~\cite{eulerInstitutionesCalculiDifferentialis1787} (a translated version is available at~\cite{eulerInstitutionesCalculiDifferentialis2019}). See the textbook of Peterson~\cite{petersenEulerianNumbers2015} for a trove of information about the Eulerian numbers and the variant defined below, including historical notes about Euler's derivation.

It is well-known that the alternating sum of Eulerian numbers $\sum_{\ell=0}^{k} (-1)^\ell \eulerian{k}{\ell} = A_k(-1)$ is the $k$th tangent number (perhaps negated). The following proof of this fact follows the structure of a survey of Foata~\cite{foataEulerianPolynomialsEuler2010}.

\[ \sum_{k \geq 1} i^{k-1} A_k(-1) \frac{\theta^k}{k!} = \frac{1}{i}\left( \sum_{k \geq 0} A_k(-1) \frac{(i \theta)^k}{k!} - A_0(-1) \right) = -i \left(A(-1, i \theta) - 1 \right) \]

From here we apply \cref{equ:euler-gen-func}.

\begin{align*} 
    \sum_{k \geq 1} i^{k-1} A_k(-1) \frac{\theta^k}{k!} &= -i \left( \frac{(-1) - 1}{-1-e^{-2i\theta}} - 1 \right) = -i \left( \frac{2}{e^{i \theta} e^{-i \theta} + e^{-i \theta} e^{-i \theta}} - 1 \right) \\&= -i \left( \frac{2 e^{i \theta}}{e^{i \theta} + e^{-i \theta}} - \frac{e^{i \theta} + e^{-i \theta}}{e^{i \theta} + e^{-i \theta}} \right) = -i \left( \frac{e^{i \theta} - e^{-i \theta}}{e^{i \theta} + e^{-i \theta}} \right) \\&= \tan \theta = \sum_{n \geq 1} \frac{\theta^{2k-1}}{(2k-1)!} T_{2k-1}
\end{align*}

Solving for $A_k(-1)$ in terms of the tangent numbers gives us
\[ A_{2k}(-1) = 0 \qquad \mbox{and} \qquad A_{2k-1}(-1) = i^{-(2k-2)}T_{2k-1} = (-1)^{k-1} T_{2k-1} \]
and therefore if $k$ is odd we have
\begin{equation} \label{equ:alternating-sum} A_k(-1) = (-1)^{\lfloor k/2 \rfloor} E_k. \end{equation}

\paragraph{Eulerian numbers of type $B_k$.} A signed permutation $\omega(\sigma, s)$ is a permutation $\sigma : [k] \rightarrow [k]$ and a sign function $s : [k] \rightarrow \{-1, 1\}$ such that $w(z) = \sigma(z) s(z)$. The Eulerian number of type $B_k$ $\eulerian{B_k}{\ell}$ is the number of signed permutations of length $k$ with $\ell$ descents. Here we consider the signed permutation to have a leading zero, so if $\omega(1) < 0$ we consider the index $0$ to have a descent.

Similarly to the standard Eulerian numbers, we associate a polynomial and an exponential generating function:
\[ B_k(t) = \sum_{\ell=0}^{k} \eulerian{B_k}{\ell} t^\ell \qquad \mbox{and} \qquad B(t, u) = \sum_{k \geq 0} B_k(t) \frac{u^k}{k!}. \]
Because the definition of descents in a signed permutation allows for a descent at index 0, the maximum number of descents is $k$. This is in contrast with the definition for permutations, where the maximum number of descents is $k-1$. This is why the range of the summation in $B_k(t)$ is different from the range in $A_k(t)$.

Brenti~\cite{brentiQEulerianPolynomialsArising1994} shows the following analogue of \cref{equ:euler-gen-func}:
\[ B(t, u) = \frac{(t-1)e^{u(t-1)}}{t-e^{2u(t-1)}}. \]

The alternating sum of Eulerian numbers of type $B_k$ are related to the secant numbers, which can be shown in a similar manner to above.

\begin{align*} \sum_{k \geq 0} B_k(-1) \frac{(i \theta)^k}{k!} &= B(-1, i \theta) = \frac{((-1)-1)e^{i \theta ((-1) -1)}}{(-1)-e^{2i\theta((-1)-1)}} = \frac{2e^{-2 i \theta}}{e^{-2 i \theta}e^{2 i \theta} + e^{-2 i \theta}e^{-2 i \theta}} = \frac{2}{e^{2 i \theta} + e^{-2 i \theta}} \\&= \sec (2\theta) = \sum_{k \geq 0} \frac{(2 \theta)^{2k}}{(2k)!} S_{2k}
\end{align*}

This yields
\[ B_{2k}(-1) = (-1)^k 2^{2k} S_{2k} \qquad \mbox{and} \qquad B_{2k+1}(-1) = 0 \]
and therefore if $k$ is even we have
\begin{equation} \label{equ:alternating-sum-b} B_{k}(-1) = (-1)^{\lfloor k/2 \rfloor} E_k. \end{equation}

\paragraph{Formulas for Eulerian numbers.} A simple formula involving a summation is well-known for the standard Eulerian numbers. We could not find a similar formula for the Eulerian numbers of type $B_k$. In the following we derive the former and show how to modify the proof to generate the latter.

The following identity is attributed to Euler~\cite{eulerRemarquesBeauRapport1768}; see~\cite{foataEulerianPolynomialsEuler2010,petersenEulerianNumbers2015} for more details.

\begin{lemma}
    $\frac{A_k(t)}{(1-t)^{k+1}} = \sum_{m \geq 0} (m+1)^k t^m.$
\end{lemma}

Using the binomial expansion for $(1-t)^{k+1}$, we get:
\begin{align*} A_k(t) 
    &= \sum_{m \geq 0} \left( (1-t)^{k+1} (m+1)^k \right) t^m 
    = \sum_{m \geq 0} (m+1)^k t^m \sum_{j=0}^{k+1} (-t)^j \binom{k+1}{j} 
    \\&= \sum_{m \geq 0} (m+1)^k \sum_{j=0}^{k+1} (-1)^j \binom{k+1}{j} t^{j + m}.
\end{align*}

Let $\ell = j + m$. Now the summation over $m \geq 0$ is a summation over $\ell - j \geq 0$ and we can rearrange to obtain:
\[ A_k(t) = \sum_{\ell \geq 0} \sum_{j = 0}^{\ell} (-1)^j \binom{k+1}{j} (\ell-j+1)^k t^\ell. \]

Using the definition $A_k(t) = \sum_{\ell=0}^{k-1} \eulerian{k}{\ell} t^\ell$ we get
\begin{equation} \label{equ:eulerian-sum} \eulerian{k}{\ell} = \sum_{j = 0}^{\ell} (-1)^j \binom{k+1}{j} (\ell-j+1)^k. \end{equation} 

A similar identity exists for the Eulerian numbers of type $B_k$ (see~\cite{petersenEulerianNumbers2015}).

\begin{lemma}
    $\frac{B_k(t)}{(1-t)^{k+1}} = \sum_{m \geq 0} (2 m+1)^k t^m.$
\end{lemma}

As the steps are exactly the same as above, we omit the details and skip to the conclusion:

\begin{equation} \label{equ:eulerian-b-sum} \eulerian{B_k}{\ell} = \sum_{j = 0}^{\ell} (-1)^j \binom{k+1}{j} (2\ell-2j+1)^k = 2^k \sum_{j = 0}^{\ell} (-1)^j \binom{k+1}{j} \left(\ell-j+\frac{1}{2}\right)^n. \end{equation} 

\paragraph{The Irwin-Hall distribution.}

The Irwin-Hall distribution (with parameter $k$) is the distribution of sums of $k$ real numbers uniformly sampled from $[0, 1]$. The cumulative distribution function of the Irwin-Hall distribution is
\[ F_k(x) = \frac{1}{k!} \sum_{j=0}^{\lfloor x \rfloor}(-1)^j\binom{k}{j}(x - j)^k \]
for $x \in [0, k]$. Tanny~\cite{tannyProbabilisticInterpretationEulerian1973} was the first to notice a connection between the Eulerian numbers and the Irwin-Hall distribution: the Eulerian numbers capture the density of the Irwin-Hall distribution on a \emph{unit interval} where the endpoints of the interval are whole numbers. It is simple to give a generalized form of this connection -- without the restriction on the endpoints of the interval -- that will also allow us to characterize the Eulerian numbers of type $B_n$.

\begin{lemma} \label{lem:irwin-hall}
    For any $x \in [1, k]$, 
    \[ k! \left( F_k(x) - F_k(x-1) \right) = \sum_{j=0}^{\lfloor x \rfloor} (-1)^j \binom{k+1}{j} (x-j)^k. \]
\end{lemma}

\begin{proof} We use the identity $\binom{k}{j} + \binom{k}{j-1} = \binom{k+1}{j}$:
    \begin{align*}
        k! \left( F_k(x) - F_k(x-1) \right) 
        &= \sum_{j=0}^{\lfloor x \rfloor} (-1)^j \binom{k}{j} (x-j)^k - \sum_{j=0}^{\lfloor x - 1 \rfloor} (-1)^j \binom{k}{j} (x-j-1)^k 
        \\&= \sum_{j=0}^{\lfloor x \rfloor} (-1)^j \binom{k}{j} (x-j)^k - \sum_{j=1}^{\lfloor x \rfloor} (-1)^{j-1} \binom{k}{j-1} (x-j)^k 
        \\&= x^k + \left[\sum_{j=1}^{\lfloor x \rfloor} (x-j)^k \left( (-1)^{j} \binom{k}{j} + (-1)^j \binom{k}{j-1} \right) \right]
        \\&= x^k + \left[\sum_{j=1}^{\lfloor x \rfloor} (x-j)^k (-1)^{j} \binom{k+1}{j} \right]
        = \sum_{j=0}^{\lfloor x \rfloor} (-1)^j \binom{k+1}{j} (x-j)^k. \qedhere
    \end{align*}
\end{proof}

The following characterizations of Eulerian numbers follow directly from \cref{equ:eulerian-sum}, \cref{equ:eulerian-b-sum}, and \cref{lem:irwin-hall}. \cref{cor:irwin-hall-eulerian} and its consequences were studied by Tanny~\cite{tannyProbabilisticInterpretationEulerian1973}.

\begin{corollary} \label{cor:irwin-hall-eulerian}
    $\eulerian{k}{\ell} = k! \left( F_n(\ell+1) - F_n(\ell) \right)$.
\end{corollary}

\begin{corollary} \label{cor:irwin-hall-eulerian-b}
    $\eulerian{B_k}{\ell} = 2^k k! \left( F_k(\ell+\frac{1}{2}) - F_k(\ell-\frac{1}{2}) \right)$.
\end{corollary}

\paragraph{Completing the proof.} Recall that $\oddevent_k$ is the event \[\sum_{j \in [k]} r_j \modulo{2} \in [-1/2,1/2).\] It is natural to write its probability in terms of the Irwin-Hall distribution by adding $1/2$ to the random variables $r_j$. We are now interested in random variables $r'_1,\dots,r'_k$ that are uniformly sampled from $[0,1)$ and the event we are interested in is \[ \sum_{j \in [k]} r'_j \modulo{2} \in \left[\frac{k-1}{2},\frac{k+1}{2}\right) \modulo{2}. \] There are four cases:
\begin{enumerate}[label=(\roman*)]
    \item \label{item:odd-odd} If $k \modulo{4} = 3$, we want the probability that the sum modulo $2$  lies in $[1,2)$:
    \[ \Pr[\oddevent_k] = F_k(k-1) - F_k(k-2) + \ldots + F_k(2) - F_k(1). \]
    Using \cref{equ:alternating-sum} and \cref{cor:irwin-hall-eulerian} this alternating sum can be rewritten as
    \begin{align}
        \Pr[\oddevent_k] 
        &= \frac{1}{2} F_k(k) - \frac{1}{2} \sum_{j=0}^{k-1} \left[ (-1)^j \left( F_k(j+1) - F_k(j) \right) \right] - \frac{1}{2} F_k(0)
        = \frac{1}{2} - \frac{1}{2} \sum_{j=0}^{k-1}\left[ (-1)^j \frac{1}{k!} \eulerian{k}{j} \right] \nonumber
        \\&= \frac{1}{2} - \frac{A_k(-1)}{2k!} 
        = \frac{1}{2} - (-1)^{\lfloor k/2 \rfloor}\frac{E_k}{2k!},\label{equ:odd_derivation}
    \end{align}
    which is equal to $\frac{1}{2} + \frac{E_k}{2k!}$.
    
    \item \label{item:odd-even} If $k \modulo{4} = 1$, we want the probability that the sum modulo $2$  lies in $[0,1)$:
    \[ \Pr[\oddevent_k] = F_k(k) - F_k(k-1) + F_k(k-2) - \ldots - F_k(2) + F_k(1). \]
    Note that $F_k(k) = 1$ and the value subtracted from $F_k(k)$ is equal to the alternating sum considered in case~\ref{item:odd-odd}. Using \cref{equ:odd_derivation} we find that this is equal to $1 - \left(\frac{1}{2} - (-1)^{\lfloor k/2 \rfloor}\frac{E_k}{2k!}\right) = \frac{1}{2} + \frac{E_k}{2k!}$.
    
    \item If $k \modulo{4} = 2$, we want the probability that the sum modulo $2$  lies in $[1/2,3/2)$:
    \[ \Pr[\oddevent_k] = F_k(k-1/2) - F_k(k-3/2) + \ldots + F_k(3/2) - F_k(1/2). \]
    We perform a similar calculation to case~\ref{item:odd-odd} with \cref{equ:alternating-sum-b} and \cref{cor:irwin-hall-eulerian-b}. Here we make use of the fact that the Irwin-Hall distribution is symmetrical: $F_k(k-j) = 1 - F_k(j)$. Additionally, we use the fact that $\eulerian{B_k}{0} = \eulerian{B_k}{k} = 1$.
    \begin{align}
        \Pr[\oddevent_k] 
        &= \frac{1}{2} F_k(k - 1/2) - \frac{1}{2} \sum_{j=1}^{k-1} \left[ (-1)^j \left( F_k(j+1/2) - F_k(j - 1/2) \right) \right] - \frac{1}{2} F_k(1/2) \nonumber
        \\&= \frac{1}{2} F_k(k) - \frac{1}{2} \sum_{j=1}^{k-1} \left[ (-1)^j \frac{1}{2^k k!} \eulerian{B_k}{j} \right] - \frac{1}{2} F_k(1/2) - \frac{1}{2} F_k(1/2) \nonumber
        \\&= \frac{1}{2} - \frac{1}{2} \sum_{j=1}^{k-1} \left[ (-1)^j \frac{1}{2^k k!} \eulerian{B_k}{j} \right] - \frac{1}{2} \frac{1}{2^k k!} \eulerian{B_k}{0} - \frac{1}{2} \frac{1}{2^k k!} \eulerian{B_k}{k} 
         \nonumber 
        \\&= \frac{1}{2} - \frac{1}{2} \sum_{j=0}^{k} \left[ (-1)^j \frac{1}{2^k k!} \eulerian{B_k}{j} \right] = \frac{1}{2} - \frac{B_k(-1)}{2^{k+1} k!}
        = \frac{1}{2} - (-1)^{k/2} \frac{E_k}{2k!},\label{equ:even-derivation}
    \end{align}
    which is equal to $\frac{1}{2} + \frac{E_k}{2k!}$.
    
    \item If $k \modulo{4} = 0$, we want the probability that the sum modulo $2$  lies in $[-1/2,1/2)$:
    \[ \Pr[\oddevent_k] = F_k(k) - F_k(k-1/2) + F_k(k-3/2) + \ldots - F_k(3/2) + F_k(1/2). \]
    The same argument from case~\ref{item:odd-even} applies: subtracting the value found in \cref{equ:even-derivation} from $F_k(k) = 1$ gives us the probability $\frac{1}{2} + \frac{E_k}{2k!}$.
\end{enumerate}
 
\printbibliography

\end{document}